\documentclass[lettersize,journal]{IEEEtran}
\usepackage{amsmath,amsfonts}
\usepackage{array}
\usepackage{textcomp}
\usepackage{stfloats}
\usepackage{url}
\usepackage{verbatim}
\usepackage{graphicx}
\usepackage{cite}
\usepackage{hyperref}
\usepackage{epsfig,amssymb,subfigure,amsmath,stfloats,latexsym}
\usepackage{pgfplots}
\usepackage{epstopdf}
\usepackage{amsfonts,amsthm}
\usepackage{multirow}
\usepackage{mathrsfs}
\usepackage{enumitem}
\usepackage[linesnumbered,ruled]{algorithm2e}
\usepackage{setspace} 

\pdfoutput=1

\newcommand{\maxGMRES}{j_m  }

\newcommand{\PSD}{N_0}
\newcommand{\fidx}{k}
\newcommand{\newidx}{r}
\newcommand{\setj}{\mathcal{C}}
\newcommand{\rdmmtx}{S}

\usepackage[flushleft]{threeparttable}
\newtheorem{definition}{Definition}
\newtheorem{thm}{Theorem}

\newtheorem{corollary}{Corollary}
\begin{document}

\title{Doubly-Iterative Sparsified MMSE Turbo Equalization for OTFS Modulation }

\author{Haotian Li,~\IEEEmembership{Student Member,~IEEE,} Qiyue Yu, \IEEEmembership{Senior Member,~IEEE}

\thanks{
The work presented in this paper was supported by the National Natural Science Foundation of China under Grand No. 62071148 and No. 62171151, partly by the Natural Science Foundation of Heilongjiang Province of China under Grand No. YQ2019F009, and partly by the Fundamental Research Funds for Central Universities under Grand No. HIT.OCEF.2021012.}

\thanks{
H. Li, and Q. Yu are with the Communication Research Center, Harbin Institute of Technology, Harbin 150001, China 
(email: lihaotian@stu.hit.edu.cn; yuqiyue@hit.edu.cn). }

}

\markboth{IEEE TRANSACTIONS ON COMMUNICATIONS, Vol. 71, No. 3, MARCH 2023} { \ldots}

\maketitle

\vspace{-0pt}
\begin{abstract}
  Currently, orthogonal time frequency space (OTFS) modulation has drawn much attention to reliable communications in high-mobility scenarios.
  This paper proposes a doubly-iterative sparsified minimum mean square error (DI-S-MMSE) turbo equalizer, which iteratively exchanges the extrinsic information between a soft-input-soft-input (SISO) MMSE estimator and a SISO decoder.
  Our proposed equalizer does not suffer from short loops and approaches the performance of the near-optimal symbol-wise maximum {\it a posteriori} (MAP) algorithm.
  To exploit the inherent sparsity of OTFS system, we resort to graph theory to investigate the sparsity pattern of the channel matrix, and propose two sparsification guidelines to reduce the complexity of calculating the matrix inverse at the MMSE estimator. 
  Then, we apply two iterative algorithms to MMSE estimation, i.e., the Generalized Minimal Residual (GMRES) and Factorized Sparse Approximate Inverse (FSPAI) algorithms.
  The former is used at the initial turbo iteration, whose global convergence is proven in our equalizer, while the latter is used at the subsequent turbo iterations with the help of our proposed guidelines. 
  Simulation results demonstrate that our equalizer can achieve a good bit-error-rate (BER) performance regardless of the fractional Doppler, and only has a linear order of complexity. Simulation codes are available to reproduce the results presented in this paper: \href{https://github.com/Alga53/DISMMSE-Turbo-Equalizer-for-OTFS}{https://github.com/Alga53/DISMMSE-Turbo-Equalizer-for-OTFS}.  
  
\end{abstract}


\vspace{-5pt}
\section{Introduction}
Next generation wireless communication systems are expected to accommodate various wireless applications in high-mobility scenarios, such as vehicle-to-vehicle (V2V), and high speed railway.
In these scenarios, the performance of conventional orthogonal frequency division multiplexing (OFDM) is greatly degraded.
To combat the Doppler spread, one of the appealing solutions is orthogonal time frequency space (OTFS)\cite{hadani2017orthogonal,hadani2018otfs}, which has great potential for supporting reliable communications in high-mobility scenarios.

OTFS multiplexes the information symbols in the delay-Doppler (DD) domain instead of the time-frequency (TF) domain.
Through a two-dimensional transform, OTFS spreads each symbol across the entire TF domain so that all symbols experience the same channel gain.
Hence, OTFS has the potential of extracting full channel diversity with a well-designed transceiver\cite{9392379}, and greatly outperforms conventional OFDM systems \cite{hadani2017orthogonal,raviteja2018interference,surabhi2019diversity,li2021performance,wei2021orthogonal}.
More importantly, the DD domain wireless channel is inherently sparse and has lower variability than that in the TF domain, which facilitates channel estimation and significantly reduces the overhead\cite{hadani2017orthogonal,raviteja2019embedded,murali2018otfs}.

However, these benefits of OTFS are obtained at the expense of prohibitive detection complexity.  
Therefore, many researchers have developed various low-complexity detectors by exploiting the sparsity of the wireless channel in the DD domain.
In \cite{tiwari2019low}, the authors have proposed a linear minimum mean square error (LMMSE) receiver with a log-linear complexity order, which utilizes the sparsity and quasi-banded structure of the matrices.
To further harvest the diversity gain, the authors of \cite{raviteja2018interference} have developed an iterative receiver based on message passing algorithm (MPA) with a linear complexity order, in which the interference is assumed to be Gaussian.
Moreover, a hybrid maximum {\it a posteriori} (MAP) and parallel interference cancellation (PIC) detection (Hybrid-MAP-PIC) algorithm has been proposed in \cite{9439819}, whose performance approaches that of the near-optimal symbol-wise MAP algorithm with low complexity.
However, due to the short loops in the factor graph that commonly exist in OTFS system, the MPA and Hybrid-MAP-PIC algorithm may fail to converge, resulting in undesirable performance degradation.
To tackle this issue, a variational Bayes approach has been developed in \cite{yuan2020simple}, which can converge regardless of the short loops in the factor graph.
Besides, the authors of \cite{yuan2021iterative} have investigated the design of OTFS detector based on approximate message passing (AMP), which handles short loops well with low complexity, and achieves a comparable bit error rate (BER) performance to the MMSE detection.
In \cite{li2021cross}, a cross domain iterative detection of OTFS system has been proposed, which passes the extrinsic information between the time and DD domain.
This detector can approach the performance of the maximum-likelihood (ML) sequence detection at the cost of a cubic complexity order.

The existing studies motivate us to achieve a desirable trade-off between the BER performance and complexity.
In this paper, we propose a doubly-iterative sparsified MMSE (DI-S-MMSE) turbo equalizer, which consists of a sparsified soft-input-soft-output (SISO) MMSE estimator and a SISO decoder.
It is observed that the overall computational complexity is dominated by computing the inverse of the covariance matrix of the received sequence for MMSE estimation.
To address this problem, our proposed equalizer sparsifies the covariance matrix so that its inverse can be derived with a linear complexity order.
More importantly, our equalizer can approach the performance of the near-optimal symbol-wise MAP detector regardless of the fractional Doppler, and does not suffer from the short loops as in the MPA \cite{raviteja2018interference} and SPA \cite{9439819} detectors.
The major contributions of this paper are summarized as follows.
\begin{enumerate}
  \vspace{-3pt}
  \item We propose a doubly-iterative sparsified MMSE turbo equalizer, which iteratively exchanges the extrinsic information between the MMSE estimator and the decoder.
  This turbo iteration is referred to as the outer iteration.
  On the other hand, to reduce the detection complexity, we use two iterative algorithms to perform MMSE estimation, i.e., the Generalized Minimal Residual (GMRES) and Factorized Sparse Approximate Inverse (FSPAI) algorithms, which are referred to as inner iteration.
  \item To fully exploit the inherent sparsity of the DD domain channel, we resort to graph theory to analyze the sparsity pattern of the channel matrix.
  Based on this, we propose two guidelines to sparsify the covariance matrix.
  In addition, considering the randomness of the channel matrix, we define the sparsity level in the form of cumulative distribution function (CDF) to measure how sparse a matrix is.
  It is validated that the sparsity level of the inverse can be greatly improved by the sparsification guidelines and thus the detection complexity is significantly reduced. 
  \item At the initial outer iteration, computing the inverse is converted into solving the equivalent sparse linear systems via GMRES.
  Moreover, we prove the global convergence of GMRES in our system.
  At the subsequent outer iterations, we employ the proposed guidelines to sparsify the covariance matrix.
  Then, using the Hermitian positive definiteness of the covariance matrix, we modify FSPAI to approximately derive the Cholesky decomposition of its inverse. 
  Simulation results show that our equalizer greatly reduces the complexity and achieves a great BER performance regardless of the fractional Doppler.
\end{enumerate}

The remainder of this paper is organized as follows.
In Section \ref{section2}, the system model of the proposed system is presented.
Section \ref{section3} provides a brief overview of the proposed equalizer.
Then, we analyze the sparsity pattern of the random sparse matrices and propose two sparsification guidelines based on graph theory in Section \ref{section4}.   
Section \ref{section5} presents the iterative algorithms.
The theoretical analysis is given in Section \ref{section6}, and simulation results are presented in Section \ref{section7}.
Finally, some conclusions are drawn in Section \ref{section8}.

\textit{Notations}: Scalars, vectors and matrices are denoted by $a$, $\mathbf{a}$ and $\mathbf{A}$ respectively. $\mathbb{B}$, $\mathbb{N}$, $\mathbb{Z}$, $\mathbb{R}$, $\mathbb{R}^{+}$ and $\mathbb{C}$ represent the set of all binary numbers, natural numbers, integers, real numbers, real positive numbers and complex numbers.
$\delta (\cdot)$ is the Dirac delta function.
$\lfloor a\rfloor$ denotes the the largest integer less than or equal to $a$. 
$\mathrm{Re}\{a\}$ and $\mathrm{Im}\{a\}$ mean the real and imaginary parts of $a$, and $a^{*}$ is the conjugate of $a$.
$\left\|{\mathbf{a}}\right\|$ represents the $\ell_2$-norm of a vector $\mathbf{a}$.
The inner product of two vectors $\mathbf{a}$ and $\mathbf{b}$ is denoted by $(\mathbf{a},\mathbf{b})$.
$\mathbf{A}(i,j)$ denotes the $(i,j)$-th element of $\mathbf{A}$.
The transpose, conjugate, conjugate transpose and inverse of a matrix are represented by $(\cdot)^\mathrm{T}$, $(\cdot)^{\ast}$, $(\cdot)^\mathrm{H}$ and $(\cdot)^{-1}$. $\mathrm{diag}[\mathbf{a}]$ denotes a diagonal matrix whose diagonal elements are $\mathbf{a}$, and $\mathrm{vec}(\cdot)$ represents the column-wise vectorization of a matrix.
$\mathcal{CN}(\mu ,\sigma^2)$ denotes the circularly symmetric complex Gaussian distribution with mean $\mu $ and variance $\sigma^2$. 
The expectation operator is denoted by $\mathbb{E}\{\cdot\}$. The covariance matrix of two vectors is denoted by $\mathrm{Cov}\{\mathbf{x},\mathbf{y}\}= \mathbb{E}\{\mathbf{xy}^{\mathrm{H}}\}-\mathbb{E\{\mathbf{x}\}}\mathbb{E}\{\mathbf{y}^{\mathrm{H}}\}$.
$\mathcal{O}(\cdot)$ is the order of asymptotic time complexity.
$|\mathcal{A}|$ is the cardinality of a set $\mathcal{A}$. 

\section{System model}\label{section2}
\subsection{OTFS Transmitter} 
The diagram of OTFS system is shown in \figurename~\ref{OTFSsys}.
The bitstream ${\bf u} \in {\mathbb B}^{K_1}$ is encoded by an encoder with a generator matrix $\mathbf{G}$ whose code rate is $R_c$, through which an encoded vector ${\bf c}=\mathbf{G}\mathbf{u} \in {\mathbb B}^{K_2}$ is obtained, where $K_1 = K_2 \cdot R_c$.
Then, $\bf c$ is interleaved to ${\bf d}\in {\mathbb B}^{K_2}$, and mapped to the symbol sequence ${\bf x} \in {\mathbb C}^{K_3}$ from the symbol alphabet $\mathcal{S}$, where $K_3=K_2/\mathrm{log}_2(|\mathcal{S}|)$.
Then $\mathbf{x}$ is arranged into an $M\times N$ DD domain matrix $\mathbf{X}_{\mathrm{DD}}$ column by column with $K_3 = MN$, where the delay and Doppler dimensions contain $M$ and $N$ samples respectively.

\begin{figure*}[t]
  \centering
  \includegraphics[width = 0.85\textwidth]{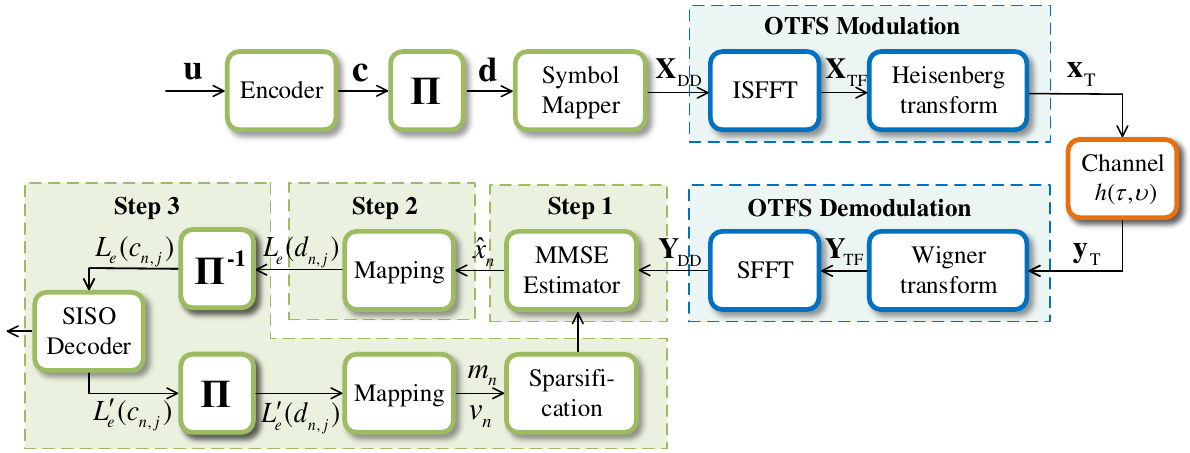}
  \caption{A diagram of OTFS transceiver, where $\mathbf{\Pi}$ and $\mathbf{\Pi}^{-1}$ are the interleaver and deinterleaver respectively. Each iteration of the proposed equalizer consists of three steps: 1) perform MMSE estimation using the {\it a priori} information; 
  2) map the estimated symbols to the extrinsic information for the decoder; 
  3) after decoding, update means and variances, and sparsify the covariance matrix of the received sequence for the MMSE estimator.}\vspace{-1em}
  \label{OTFSsys}
\end{figure*}

Next, the inverse symplectic finite Fourier transform (ISFFT) performs an $M$-point FFT along the delay dimension and an $N$-point IFFT along the Doppler dimension to derive the TF domain transmit symbols matrix $\mathbf{X}_{\mathrm{TF}}\in \mathbb{C}^{M\times N}$ as \cite{raviteja2018practical}
\begin{equation} \label{ISFFT}
  \mathbf{X}_{\mathrm{TF}} = \mathbf{F }_M \mathbf{X}_{\mathrm{DD}} \mathbf{F  }_N^{\mathrm{H}},
\end{equation}
where $\mathbf{F  }_M$ denotes the $M$-point DFT matrix and $\mathbf{F  }_N^\mathrm{H}$ denotes the $N$-point IDFT matrix. 
Here, $NT$ and $M\Delta f$ represent the time duration and bandwidth respectively, where the subcarrier spacing is $\Delta f$ (Hz) and the time interval is $T$ (seconds).

After the ISFFT, the Heisenberg transform performs an $M$-point IFFT along each column of $\mathbf{X}_{\mathrm{TF}}$ and applies the pulse shaping waveform $g_{\mathrm{tx}}(t)$ of duration $[0,T)$. 
Then, the obtained time domain transmit matrix $\mathbf{X}_{\mathrm{T}}$ can be denoted by \cite{raviteja2018practical}
\begin{equation}
  \mathbf{X}_{\mathrm{T}} = \mathbf{G  }_{\mathrm{tx}} \mathbf{F  }_M^{\mathrm{H}} \mathbf{X}_{\mathrm{TF}}= \mathbf{G  }_{\mathrm{tx}} \mathbf{X}_{\mathrm{DD}} \mathbf{F  }_N^{\mathrm{H}},
\end{equation}
where the diagonal elements of $\mathbf{G  }_{\mathrm{tx}}$ are samples of $g_{\mathrm{tx}}(t)$ with sampling interval $T/M$, i.e., $\mathbf{G  }_{\mathrm{tx}}=\mathrm{diag}\left[g_{\mathrm{tx}}(0), g_{\mathrm{tx}}(T/M),\dots, g_{\mathrm{tx}}((M-1)T/M)\right]\in \mathbb{C}^{M\times M}$. 
Without loss of generality, this paper only considers the rectangular waveform, where $\mathbf{G  }_{\mathrm{tx}}$ is simplified as the identity matrix $\mathbf{I}_M$.
Nonetheless, some other well-designed windows in \cite{wei2021transmitter} can also be applied.
Then, the transmit signal can be vectorized as
$
  \mathbf{x}_{\mathrm{T}} = (\mathbf{F  }_N^{\mathrm{H}}\otimes \mathbf{G  }_{\mathrm{tx}} )\mathbf{x} = (\mathbf{F  }_N^{\mathrm{H}}\otimes \mathbf{I}_M)\mathbf{x},
$
where $\mathbf{x}_{\mathrm{T}}= \mathrm{vec}(\mathbf{X}_{\mathrm{T}})\in \mathbb{C}^{MN}$ and $\otimes$ represents the kronecker product. 
Such a vector of length $MN$ is termed an OTFS frame\cite{hadani2017orthogonal}. 
A cyclic prefix (CP) of length $L$ is appended to each OTFS frame, and after digital to analog conversion, the time domain signal $x_{\mathrm{T}}(t)$ is transmitted into the time-varying channel.

\vspace{-0pt}
\subsection{Channel Model}
The channel impulse response can be characterized as $h(\tau,\nu)$ with delay $\tau$ and Doppler $\nu$
\vspace{-0pt}
\begin{equation} \label{h}
  h(\tau, \nu)=\sum_{p=1}^{P} h_{p} \delta\left(\tau-\tau_{p}\right) \delta\left(\nu-\nu_{p}\right),
  \vspace{-0pt}
\end{equation}
where $h_p,\tau_p$ and $\nu_p$ represent the channel gain, delay and Doppler shift of the $p$-th path for $p=1,\dots,P$, and $P$ is the number of paths\cite{raviteja2018interference}.
Here, the delay and Doppler shift are given by 
\vspace{-0pt}
\begin{equation}
  \tau_p = \frac{l_p}{M\Delta f}, ~~\nu_p = \frac{k_p+\kappa_p}{NT},
  \vspace{-0pt}
\end{equation} 
where the delay indices $l_p$ and Doppler shift indices $k_p$ are integers, i.e., $k_p, l_p\in\mathbb{Z}$, and the term $\kappa_p \in [-1/2, 1/2]$ is the fractional Doppler.
For conventional wide-band systems, the sampling time $1/(M\Delta f)$ is sufficiently small so that the delay taps can be approximated as integers\cite{raviteja2018interference}. However, due to the limit of latency, $N$ cannot be chosen arbitrarily large.
Therefore, the fractional Doppler may exist in practical systems.
In the following, we assume $N$ is sufficiently large so that the fractional Doppler can be neglected.
This facilitates us to understand the sparsity characteristics of the channel matrix.
Then in the simulation part, we will demonstrate that our proposed equalizer can also handle the fractional Doppler well with low complexity. 
 
The aforementioned length of the CP should be larger than the maximum delay, i.e., $L> \tau_{\mathrm{max}}M\Delta f$. 
According to \cite{hadani2017orthogonal}, the received signal $y_{\mathrm{T}}(t)$ is
\vspace{-0pt}
\begin{equation} \label{rt}
  y_{\mathrm{T}}(t)=\iint h(\tau, \nu) e^{j 2 \pi \nu(t-\tau)} x_{\mathrm{T}}(t-\tau) \mathrm{d} \tau \mathrm{d} \nu+z_{\mathrm{T}}(t),
  \vspace{-0pt}
\end{equation}
where $z_{\mathrm{T}}(t)$ is the additive white Gaussian noise (AWGN) in the time domain, which is distributed as $\mathcal{CN}(0,\PSD)$.
After CP is removed, $y_{\mathrm{T}}(t)$ is sampled with interval $T/M$. By substituting (\ref{h}), the discrete received signal is obtained as  \cite{raviteja2018practical}
\vspace{-0pt}
\begin{equation}
  \label{discreteR}
  y_{\mathrm{T}}(n) = \sum_{p=1}^P h_p e^{j2\pi \frac{k_p(n-l_p)}{MN}}x_{\mathrm{T}}([n-l_p]_{MN})+z_{\mathrm{T}}(n),
  \vspace{-0pt}
\end{equation}
where $n=1,\dots,MN$ and $[\cdot]_{MN}$ represents the modulo-$MN$ operation.
From (\ref{discreteR}), the time domain effective channel matrix $\mathbf{H}_{\mathrm{T}}\in \mathbb{C}^{MN\times MN}$ can be written as \cite{raviteja2018practical}
\vspace{-0pt}
\begin{equation} \label{HT}
  \mathbf{H}_{\mathrm{T}} = \sum_{p=1}^{P} h_p \mathbf{\Pi}^{l_p} \mathbf{\Delta}^{k_p},
  \vspace{-0pt}
\end{equation}
where $\mathbf{\Pi}$ is a permutation matrix (forward cyclic shift) and represents the effect of delay
\begin{equation}
  \setlength{\arraycolsep}{2pt}
  \renewcommand{\arraystretch}{0.5}
  \boldsymbol{\Pi}_{MN}=\left[\begin{array}{cccc}
      0      & \cdots & 0      & 1      \\
      1      & \ddots & 0      & 0      \\
      \vdots & \ddots & \ddots & \vdots \\
      0      & \cdots & 1      & 0
    \end{array}\right]_{MN\times MN},
\end{equation}
and $\mathbf{\Delta}$ models the Doppler effect and introduces an extra frequency shift, which is an $MN\times MN$ diagonal matrix as
$
  \mathbf{\Delta} = \mathrm{diag}[\omega^0, \omega^1,\cdots,\omega^{MN-1}]
$
with $\omega= e^{\frac{j2\pi}{MN}}$. In (\ref{HT}), each path incurs an $l_p$-step cyclic time shift, represented by $\mathbf{\Pi}^{l_p}$, along with a frequency shift $k_p$ of the transmit signal $\mathbf{x}_{\mathrm{T}}$, represented by $\mathbf{\Delta}^{k_p}$. In this way, the received signal $\mathbf{y}_{\mathrm{T}}\in\mathbb{C}^{MN}$ of (\ref{discreteR}) is given by
\vspace{-0pt}
\begin{equation}
  \mathbf{y}_{\mathrm{T}}=\mathbf{H}_{\mathrm{T}}\mathbf{x}_{\mathrm{T}}+\mathbf{z}_{\mathrm{T}},
  \vspace{-0pt}
\end{equation}
where $\mathbf{z}_{\mathrm{T}}\in \mathbb{C}^{MN}$ is the AWGN vector.

\vspace{-0pt}
\subsection{OTFS Receiver}
At the receiver, $\mathbf{y}_{\mathrm{T}}$ is rearranged into an $M\times N$ matrix $\mathbf{Y}_{\mathrm{T}}$ column by column. 
Similarly, we consider rectangular received pulse shaping waveform. 
Then, an $M$-point FFT operation is applied to each column of $\mathbf{Y}_{\mathrm{T}}$ to obtain the TF domain received signal $\mathbf{Y}_{\mathrm{TF}}\in \mathbb{C}^{M\times N}$ as 
\vspace{-0pt}
\begin{equation} \label{YTF}
  \mathbf{Y}_{\mathrm{TF}} = \mathbf{F  }_M \mathbf{G  }_{\mathrm{rx}} \mathbf{Y}_{\mathrm{T}} = \mathbf{F  }_M \mathbf{I}_M \mathbf{Y}_{\mathrm{T}}.
  \vspace{-0pt}
\end{equation}

Afterwards, the SFFT is performed, which consists of an $M$-point IFFT of columns and an $N$-point FFT of rows. Thus, by substituting (\ref{YTF}), the DD domain received matrix $\mathbf{Y}_{\mathrm{DD}}$ is yielded as
$
  \mathbf{Y}_{\mathrm{DD}} = \mathbf{F  }_M^{\mathrm{H}} \mathbf{Y}_{\mathrm{TF}} \mathbf{F  }_N = \mathbf{I}_M \mathbf{Y}_{\mathrm{T}} \mathbf{F  }_N
$,
whose vectorized form is
$
  \mathbf{y} = \mathrm{vec}(\mathbf{Y_{\mathrm{DD}}}) = (\mathbf{F  }_N\otimes \mathbf{I}_M)\mathbf{y}_{\mathrm{T}}.
$
Then, the transmitted DD domain vector $\mathbf{x}$ and the received DD domain vector $\mathbf{y}$ are related as
\vspace{-0pt}
\begin{equation}
  \begin{aligned}
 \label{yHDDx}
    \mathbf{y}  &=\left(\mathbf{F  }_{N} \otimes \mathbf{I}_M\right) \mathbf{H}_{\mathrm{T}}\left(\mathbf{F  }_{N}^{\mathrm{H}} \otimes \mathbf{I}_M\right) \mathbf{x}+\left(\mathbf{F  }_{N} \otimes \mathbf{I}_M\right) \mathbf{z}_{\mathrm{T}} \\
    &=\mathbf{H}_{\mathrm{DD}} \mathbf{x}+\mathbf{z},
    \vspace{-0pt}
  \end{aligned}
\end{equation}
where $\mathbf{z}$ is the DD domain noise vector, and the effective DD domain channel matrix is
\vspace{-0pt}
\begin{equation} \label{HDDp}
  \mathbf{H}_{\mathrm{DD}}=(\mathbf{F  }_{N} \otimes \mathbf{I}_M) \mathbf{H}_{\mathrm{T}}(\mathbf{F  }_{N}^{\mathrm{H}} \otimes \mathbf{I}_M)=\sum_{p=1}^{P} h_p \mathbf{H}_{\mathrm{DD}}^{(p)},
  \vspace{-0pt}
\end{equation}
where the $p$-th component $\mathbf{H}_{\mathrm{DD}}^{(p)}\in\mathbb{C}^{MN\times MN}$ is caused by the $p$-th path. 
From \cite{raviteja2018practical}, $\mathbf{H}_{\mathrm{DD}}^{(p)}$ has been derived in (\ref{Hp}) as shown at the top of the next page, where $0\leq a,b\leq MN-1$, $k_a=\lfloor a/M \rfloor$ and $l_a=a-k_a M$. 
Apparently, $\mathbf{H}^{(p)}_{\mathrm{DD}}$ contains only one nonzero element in each row and column, and thus
$\mathbf{H}_{\mathrm{DD}}$ is a sparse matrix with $P$ nonzero elements in each row and column. 
In what follows, the received sequence $\mathbf{y}$ is fed into the turbo equalizer to recover the information bits. 
In practice, the size of an OTFS frame, i.e., $MN$ is usually very large. 
Therefore, a low-complexity equalizer is necessary.

\begin{figure*}[t]
  \centering
  \vspace{-0pt}
\begin{equation} \label{Hp}
  \mathbf{H}^{(p)}_{\mathrm{DD}}(a, b)=
  \begin{cases} 
    \omega  ^{k_{p}\left[l_a-l_{p}\right]_{M}-k_aM}, & \text { if } b=\left[l_a-l_{p}\right]_{M}+M\left[k_a-k_{p}\right]_{N} \text { and } l_a<l_{p}      \\
    \omega^{k_{p}\left[l_a-l_{p}\right]_{M}},      & \text { if } b=\left[l_a-l_{p}\right]_{M}+M\left[k_a-k_{p}\right]_{N} \text { and } l_a  \geq l_{p} \\
    0,                                                        & \text { otherwise }
  \end{cases}
  \vspace{-0pt}
\end{equation}
\noindent\rule[0.25\baselineskip]{\textwidth}{0.3pt}
\end{figure*}

\section{Doubly-iterative sparsified MMSE turbo equalizer}
\label{section3}
In this section, we propose a doubly-iterative sparsified MMSE (DI-S-MMSE) turbo equalizer, which is composed of a SISO MMSE estimator and a SISO decoder.
The proposed equalizer has a doubly iterative structure, i.e., the outer and inner iteration.
The outer iteration is operated by passing soft information between the MMSE estimator and the decoder iteratively; the inner iteration is performed inside the MMSE estimator to reduce the computational cost by means of two iterative algorithms, i.e., the GMRES algorithm at the initial outer iteration and the FSPAI algorithm at the subsequent outer iterations. 
% More importantly, to fully exploit the inherent sparsity of the channel, we further sparsify the covariance matrix of the received sequence to facilitate MMSE estimation so that the detection complexity can be greatly reduced with little performance degradation. 
This section mainly focuses on the outer iteration, and the following sections will investigate the sparsification methods and the iterative algorithms.

\vspace{-10pt}
\subsection{An overview of the outer iteration}
As shown in \figurename~\ref{OTFSsys}, the outer iteration of our proposed equalizer consists of three steps: 1) perform MMSE estimation using the {\it a priori} information; 2) calculate the extrinsic information for the decoder; 3) after decoding, update means and variances, and sparsify the covariance matrix for the MMSE estimator.
In this paper, we use quadrature phase shift keying (QPSK) with alphabet $\mathcal{S}=\{ s_1=(1+i)/\sqrt{2},s_2=(1-i)/\sqrt{2},s_3=(-1+i)/\sqrt{2},s_4=(-1-i)/\sqrt{2} \}$.
Here, $s_k$ for $k=1,\dots,4$ represents the $k$-th constellation, whose corresponding two bits $\mathbf{s}_k=[s_{k,1}, s_{k,2}]$ are $[0,0],[0,1],[1,0]$ and $[1,1]$ respectively.

% To sum up, each outer iteration includes three steps: 
% 1) perform MMSE estimation using the {\it a priori} information; 
% 2) map the estimated symbols to the extrinsic information for the decoder; 
% 3) after decoding, update means and variances for the MMSE estimator, and sparsify the covariance matrix.

The first step is to estimate the transmitted symbols from the received sequence $\mathbf{y}$ by using the MMSE estimator. 
The input of the MMSE estimator consists of the mean $m_n=\mathbb{E}\{x_n\}$ and the variance $v_n=\mathbb{E}\{|x_n-m_n|^2\}$ of each symbol $x_n$ for $n=1,\dots, MN$. 
To minimize the mean square error (MSE) $\mathbb{E}\{\left\| \mathbf{x-\hat{x}} \right\|^2 \}$, the estimator derives the estimated symbol $\hat{x}_n$ from $\mathbf{y}$ as \cite{poor2013introduction}
\vspace{-0pt}
\begin{equation} \label{MMSE}
  \hat{x}_n = m_n+\mathrm{Cov}\{x_n,\mathbf{y}\}\mathrm{Cov}\{\mathbf{y},\mathbf{y}\}^{-1}(\mathbf{y}-\mathbb{E}\{\mathbf{y}\}).
  \vspace{-0pt}
\end{equation}
For notational convenience, the covariance matrix of the received sequence $\mathbf{y}$ is denoted by $\mathbf{A}=\mathrm{Cov}\{\mathbf{y},\mathbf{y}\} $, and computing its inverse largely dominates the complexity of our equalizer.

Next, the estimated symbol $\hat{x}_n$ is mapped to its corresponding log-likelihood ratio (LLR) for the decoder. 
To harvest the potential full channel diversity gain, the {\it a posteriori} probability of each bit with respect to the entire received sequence $\mathbf{y}$ is desirable, i.e., $P(d_{n,j}|\mathbf{y})$, where
$\mathbf{d}_n=[d_{n,1},d_{n,2}]$ denotes the interleaved bits corresponding to the symbol $x_n\in \mathcal{S}$ and $j=1,2$. 
However, computing $P(d_{n,j}|\mathbf{y})$ is extremely time-consuming, especially when $M$ and $N$ are large. 
Therefore, \cite{tuchler2002minimum} inventively introduces the {\it a posteriori} probability with respect to the estimated symbol $P(d_{n,j}|\hat{x}_n)$ to approximate $P(d_{n,j}|\mathbf{y})$.
The LLR of $P(d_{n,j}|\hat{x}_n)$ is defined as
\vspace{-2pt}
\begin{equation}\small
  \begin{aligned}
    &L(d_{n,j}|\hat{x}_n)  \triangleq \mathrm{ln}\frac{P(d_{n,j}=0|\hat{x}_n)}{P(d_{n,j}=1|\hat{x}_n)} \\           
    =&\mathrm{ln}\frac{\sum_{\forall \mathbf{d}_n:d_{n,j}=0} p(\hat{x}_n|\mathbf{d}_n)P(\mathbf{d}_n) }
    {\sum_{\forall \mathbf{d}_n:d_{n,j}=1} p(\hat{x}_n|\mathbf{d}_n)P(\mathbf{d}_n)}                                             \\
      =&\mathrm{ln}\frac{\sum_{\forall \mathbf{d}_n:d_{n,j}=0} p(\hat{x}_n|\mathbf{d}_n) \prod_{\forall j^{\prime}:j^{\prime}\neq j} P(d_{n,j^{\prime}}) }
    {\sum_{\forall \mathbf{d}_n:d_{n,j}=1} p(\hat{x}_n|\mathbf{d}_n)  \prod_{\forall j^{\prime}:j^{\prime}\neq j} P(d_{n,j^{\prime}})} +\mathrm{ln}\frac{P(d_{n,j}=0)}{P(d_{n,j}=1)}, \label{Ext}
  \end{aligned}
  \vspace{-0pt}
\end{equation}
where $j^{\prime}=1,2$.
% The QPSK symbol alphabet $\mathcal{S}= \{ s_1, s_2, s_3, s_4 \}$ is exhibited in Table \uppercase\expandafter{\romannumeral1}, where $s_k$ denotes the $k$-th constellation for $k=1,\dots,4$, whose corresponding two bits are $\mathbf{s}_k=[s_{k,1}, s_{k,1}]$, and each transmit symbol $x_n\in \mathcal{S}$.
The first term of (\ref{Ext}) is defined as the extrinsic LLR $L_e(d_{n,j})$, and the second term is defined as the {\it a priori} LLR $L(d_{n,j})$, which is derived from the decoder. 
It is worth noting that $L_e(d_{n,j})$ and $L(d_{n,j})$ should be strictly independent.

Last, the estimator outputs the extrinsic information, which serves as the {\it a priori} information of the decoder. 
Thereafter, the decoder further mitigates the interference, and feeds the extrinsic information back to the estimator, from which new means and variances are obtained. 
Before we start another outer iteration, the covariance matrix $\mathbf{A}$ is sparsified to reduce the detection complexity.
In the following, we will give a detailed description of these three steps.

% To sum up, each outer iteration includes three steps: 
% 1) perform MMSE estimation using the {\it a priori} information; 
% 2) map the estimated symbols to the extrinsic information for the decoder; 
% 3) after decoding, update means and variances for the MMSE estimator, and sparsify the covariance matrix. 

% \vspace{-0.25in}
% \begin{table}[h]
%   \small
%   \begin{spacing} {1.2}
%   \centering
%   \label{alphabet}
%   \caption{QPSK symbol alphabet}
%   \begin{tabular} {|c|c|c|c|c|}
%     \hline
%     $k$                   & 1                         & 2                         & 3                         & 4                         \\
%     \hline
%     ($s_{k,1}$,$s_{k,2}$) & (0,0)                     & (0,1)                     & (1,0)                     & (1,1)                     \\
%     \hline
%     $s_k$                 & ${(+1+i)}/{\sqrt{2}}$ & $\{{(+1-i)}/{\sqrt{2}}$ & ${(-1+i)}/{\sqrt{2}}$ & ${(-1-i)}/{\sqrt{2}}$ \\
%     \hline
%   \end{tabular}
% \end{spacing}
% \end{table} 
% \vspace{-0.25in}

\vspace{-10pt}
\subsection{Step 1: MMSE Estimation Using the {\it A Priori} Information}
At the initial iteration, it is assumed that $m_n=0$ and $v_n=1$, while at the subsequent iterations, means and variances are provided by the decoder.
By substituting (\ref{yHDDx}), we have
\vspace{-0pt}
\begin{subequations}
  \begin{align}
    \label{Covx} \mathrm{Cov} \{\mathbf{x}, \mathbf{x}\} & = \mathrm{diag} [v_1,\dots,v_{MN}] \triangleq \mathbf{V},                              \\
    \label{Covy} \mathrm{Cov} \{\mathbf{y}, \mathbf{y}\} & = \mathbf{A}= \mathbf{H}_{\mathrm{DD}}\mathbf{V}\mathbf{H}_{\mathrm{DD}}^{\mathrm{H}}+\PSD \mathbf{I},                                                                    \\
    \mathrm{Cov} \{\mathbf{z}, \mathbf{z}\}              & =
    \label{Covz}\left(\mathbf{F  }_{N} \otimes \mathbf{I}_M\right) \mathrm{Cov} \{\mathbf{z}_{\mathrm{T}}, \mathbf{z}_{\mathrm{T}}\} \left(\mathbf{F  }_{N}^{\mathrm{H}} \otimes \mathbf{I}_M\right) = \PSD \mathbf{I}, 
  \end{align}
\end{subequations}
\vspace{-10pt}

\noindent where (\ref{Covx}) is obtained due to the assumption that the transmitted symbols are independent, which approximately holds due to the interleaver; (\ref{Covz}) is derived from $\mathrm{Cov} \{\mathbf{z}_{\mathrm{T}}, \mathbf{z}_{\mathrm{T}}\} = \PSD \mathbf{I}$ and $(\mathbf{F  }_{N} \otimes \mathbf{I}_M)\left(\mathbf{F  }_{N}^{\mathrm{H}} \otimes \mathbf{I}_M\right)=\mathbf{I}$ where the subscript of $\mathbf{I}_{MN}$ can be omitted without ambiguity.
Moreover, we derive
$
  \mathrm{Cov}\{x_n,\mathbf{y}\} = \mathrm{Cov}\{x_n,\mathbf{x}\}\mathbf{H}_{\mathrm{DD}}^{\mathrm{H}} = v_n\mathbf{h}_n^{\mathrm{H}},
$
where $\mathbf{h}_{n}$ represents the $n$-th column of $\mathbf{H}_{\mathrm{DD}}$. 
By substituting (\ref{Covy}) and $\mathrm{Cov}\{x_n,\mathbf{y}\}$ into (\ref{MMSE}), the estimated symbol $\hat{x}_n$ is derived as
\vspace{-0pt}
\begin{equation} \label{MMSE2}
  \hat{x}_n = m_n+v_n\mathbf{h}_n^{\mathrm{H}}
  \mathbf{A}^{-1}
  (\mathbf{y}-\mathbf{H}_{\mathrm{DD}}\mathbf{m}),
  \vspace{-0pt}
\end{equation}
where $\mathbf{m}=[m_1,\dots,m_{MN}]^{\mathrm{T}}$ denotes the vector of means.

Note that the derivation of $\hat{x}_n$ relies on the {\it a prior} knowledge $L(d_{n,j})$, so that the resultant $L_e(d_{n,j})$ derived from $\hat{x}_n$ is not independent of $L(d_{n,j})$. 
Thus, (\ref{MMSE2}) is modified as \cite{tuchler2002minimum}
\vspace{-0pt}
\begin{equation} \label{MMSE3}
  \hat{x}_n = 0+1\cdot\mathbf{h}_n^{\mathrm{H}}
  \left(\mathbf{A}+(1-v_n) \mathbf{h}_n \mathbf{h}_n^{\mathrm{H}} \right)^{-1}
  (\mathbf{y}-\mathbf{H}_{\mathrm{DD}}\mathbf{m}+m_n\mathbf{h}_n),
  \vspace{-0pt}
\end{equation}
where $\hat{x}_n$ excludes its own {\it a priori} knowledge by setting $m_n=0$ and $v_n=1$. 
To avoid computing the inverse in (\ref{MMSE3}) repeatedly for each $\hat{x}_n$, the Woodbury matrix identity \cite{golub1996matrix} is utilized as
\vspace{-0pt}
\begin{equation}
  \begin{aligned} \label{Woodbury}
      &\mathbf{h}_n^{\mathrm{H}}\left(\mathbf{A}+(1-v_n) \mathbf{h}_n \mathbf{h}_n^{\mathrm{H}} \right)^{-1}\\=&                                                                                                                        
      \mathbf{h}_n^{\mathrm{H}}\mathbf{A}^{-1}-\mathbf{h}_n^{\mathrm{H}}\mathbf{A}^{-1} \mathbf{h}_n \left[(1-v_n)^{-1}+\mathbf{h}_n^{\mathrm{H}} \mathbf{A}^{-1} \mathbf{h}_n\right]^{-1}\mathbf{h}_n^{\mathrm{H}}\mathbf{A}^{-1} \\
      =&[1+(1-v_n)\xi _n]^{-1}\mathbf{h}_n^{\mathrm{H}}\mathbf{A}^{-1},
  \end{aligned}
  \vspace{-0pt}
\end{equation}
where $\xi _n= \mathbf{h}_n^{\mathrm{H}} \mathbf{A}^{-1} \mathbf{h}_n$.
Then, by substituting (\ref{Woodbury}) into (\ref{MMSE3}), the estimated symbol is
\vspace{-0pt}
\begin{equation} \label{MMSELAST}
  \hat{x}_n = \frac{ \mathbf{h}_n^{\mathrm{H}} \mathbf{A}^{-1} (\mathbf{y}-\mathbf{H}_{\mathrm{DD}}\mathbf{m})+m_n\xi _n } {1+(1-v_n)\xi _n}.
\end{equation}
From (\ref{MMSELAST}), we observe that $\mathbf{A}^{-1}$ is used twice, i.e., $\mathbf{h}_n^{\mathrm{H}} \mathbf{A}^{-1} (\mathbf{y}-\mathbf{H}_{\mathrm{DD}}\mathbf{m})$ and $\xi_n$, which dominates the complexity of MMSE estimation.

\vspace{-5pt}
\subsection{Step 2: Calculate the Extrinsic Information for the Decoder}
To calculate $L_e(d_{n,j})$ in (\ref{Ext}), $p(\hat{x}_n|\mathbf{d}_n=\mathbf{s}_k)$ can be assumed to be Gaussian distribution as
$p(\hat{x}_n|\mathbf{d}_n=\mathbf{s}_k) = p(\hat{x}_n|x_n=s_k) =
  1/(\pi \sigma_{n,k}^2)\mathrm{exp}(-|\hat{x}_n-\mu _{n,k}|^2/\sigma_{n,k}^2)$.
Here, $\mu _{n,k} \triangleq \mathbb{E}\{ \hat{x}_n|x_n=s_k \}$ and $\sigma_{n,k}^2\triangleq \mathbb{E}\{ |\hat{x}_n-\mu _{n,k}|^2 \mid x_n=s_k\}$ are the mean and variance of the estimation error $e_{n,k}=\hat{x}_n-s_k$. 
According to (\ref{MMSELAST}), $\mu _{n,k}$ and $\sigma_{n,k}^2$ are given by
\vspace{-5pt}
\begin{equation} \label{mu}
  \begin{aligned}
    \mu _{n,k}  = &\frac{
      \left[\mathbf{h}_n^{\mathrm{H}} \mathbf{A}^{-1} \left( \mathbb{E}\{\mathbf{y}|x_n=s_k\} - \mathbf{H}_{\mathrm{DD}}\mathbf{m}\right)+m_n\xi _n\right]
    }{1+(1-v_n)\xi _n}             \quad\quad\quad               \\
     & = \frac{1}{1+(1-v_n)\xi _n} \mathbf{h}_n^{\mathrm{H}} \mathbf{A}^{-1}\mathbf{h}_n s_k = \frac{\xi _ns_k}{1+(1-v_n)\xi _n},
  \end{aligned}
\end{equation}
\vspace{-5pt}
\begin{equation} \label{sigma}
  \begin{aligned}
    \sigma_{n,k}^2 & = \frac{
      \mathbf{h}_n^{\mathrm{H}} \mathbf{A}^{-1}
    \mathrm{Cov}\{\mathbf{y},\mathbf{y}|x_n=s_k\}  (\mathbf{A}^{-1})^{\mathrm{H}} \mathbf{h}_n
    }{\left[1+(1-v_n)\xi _n\right]^2}   \\
    & =\frac{\xi _n^{\ast}-v_n\xi _n \xi _n^{\ast} }{\left[1+(1-v_n)\xi _n\right]^2}
    \overset{(a)}{=} \frac{\xi _n(1-v_n\xi _n)}{\left[1+(1-v_n)\xi _n\right]^2},
  \end{aligned}
  \quad\quad\quad\quad\quad~
\end{equation}
\vspace{-5pt}

\noindent where (a) follows from the fact that $\mathbf{A}$ is Hermitian, and thus $\xi _n$ is a real number. By substituting (\ref{mu}) and (\ref{sigma}) into the extrinsic LLR of (\ref{Ext}), $L_e(d_{n,j})$ for $j=1,2$ are derived as \cite{tuchler2002minimum}
\vspace{-0pt}
\begin{subequations}
\begin{align}
  \label{Le}
    L_e(d_{n,1})  &= \frac{\sqrt{8} [1+(1-v_n)\xi_n] \mathrm{Re}\{ \hat{x}_n\}} {1-v_n \xi _n}, \\
    L_e(d_{n,2})  &= \frac{\sqrt{8} [1+(1-v_n)\xi_n] \mathrm{Im}\{ \hat{x}_n\}} {1-v_n \xi _n}.
\end{align}
\end{subequations}
Then, $L_e(c_{n,j})$ is deinterleaved into $L_e(c_{n,j})$, and serves as the {\it a priori} information of the decoder.

\vspace{-5pt}
\subsection{Step 3: Update Means and Variances and Sparsify the Covariance Matrix}
After decoding, the decoder outputs new extrinsic information $L_e^{\prime}(c_{n,j})$. Through interleaver, $L_e^{\prime}(d_{n,j})$ is obtained and fed back to the estimator.  Considering the definition of LLR, $P(x_n=s_k)$ can be readily derived as
\vspace{-3pt}
\begin{equation}
  \begin{aligned}
    &P(x_n=s_k) =  \prod_{j=1}^{2}P(d_{n,j}=s_{k,j})                                  
    \\=&             \prod_{j=1}^{2} \frac{1}{2} \cdot \left[1-\mathrm{sgn}(s_{k,j}-\frac{1}{2}) \cdot \mathrm{tanh}(\frac{L_e^{\prime}(d_{n,j})}{2})\right],
    \vspace{-0pt}
  \end{aligned}
\end{equation}
where $\mathrm{sgn}(\cdot)$ is the sign function.
From $\mathcal{S}$, new means and variances are obtained as
\vspace{-0pt}
\begin{equation*}
  \begin{aligned}
    m_n & = \sum_{k=1}^{4} s_k\cdot P(x_n=s_k)\\
    &= \frac{1}{\sqrt{2}} \left[ \mathrm{tanh}(\frac{L_e^{\prime}(d_{n,1})}{2})+i\cdot \mathrm{tanh}(\frac{L_e^{\prime}(d_{n,2})}{2}) \right], \\
    v_n & = \sum_{k=1}^{4} |x_n-m_n|^2 P(x_n=s_k)= 1-|m_n|^2.
  \end{aligned}
\end{equation*}
\vspace{-0pt}

\noindent
At the end of an outer iteration, we expect to sparsify $\mathbf{A}$ for ease of computing $\mathbf{A}^{-1}$.
To this end, we resort to graph theory to propose some guidelines as shown in the next section. 

\vspace{-0pt}
\section{Sparsification of the covariance matrix based on graph theory}
\label{section4}
As mentioned before, high complexity of the proposed equalizer is mainly caused by computing $\mathbf{A}^{-1}$ in MMSE estimation, since $\mathbf{A}$ is an $MN\times MN$ random matrix, and $MN$ is typically large. 
For example, if we have $M=64$ and $ N=32$, directly calculating $\mathbf{A}^{-1}$ takes roughly $(MN)^3=2^{33}$ complex multiplications. 
According to (\ref{Covy}), $\mathbf{A}$ is a function of the random sparse matrix $\mathbf{H}_{\mathrm{DD}}$ so that $\mathbf{A}$ is also a random sparse matrix. 
% If the inherent sparsity is well utilized, the computational complexity of equalizer can be greatly reduced.
Unfortunately, the sparsity of $\mathbf{A}$ cannot guarantee that $\mathbf{A}^{-1}$ is also sparse.
On the contrary, the inverse $\mathbf{A}^{-1}$ is dense in most cases.
% In other words, even though $\mathbf{A}$ is sparse, deriving $\mathbf{A}^{-1}$ precisely requires a prohibitive complexity.
Hence, in this section, we resort to graph theory to provide some insight into the sparsity pattern of $\mathbf{A}$, and propose some guidelines to sparsify $\mathbf{A}$ so that $\mathbf{A}^{-1}$ can be approximated as a sparse matrix and thus the detection complexity can be greatly reduced.

\vspace{-0pt}
\subsection{Sparsification Guidelines Based on Graph Theory}
Without loss of generality, we first introduce an arbitrary random sparse matrix $\mathbf{S}$, and present two important definitions to consider its sparsity and randomness respectively.
Then using graph theory, we propose two sparsification guidelines to sparsify $\mathbf{S}$.
To begin with, we define the sparsity pattern of $\mathbf{S}$ to describe the locations of nonzero elements.

\vspace{-3pt}
\begin{definition} 
(Sparsity pattern)
For an arbitrary random sparse matrix $\mathbf{S}\in\mathbb{C}^{N_S\times N_S}$, its sparsity pattern is denoted by a Boolean matrix $\mathbf{\Upsilon_{\rdmmtx}}\in\mathbb{B}^{N_{\rdmmtx}\times N_{\rdmmtx}}$ to extract the locations of nonzero elements, where $\mathbf{\Upsilon_{\rdmmtx}}(i,j)=1$ if and only if $\mathbf{\rdmmtx}(i,j)\neq 0$.
\end{definition}
\vspace{-3pt}

In our system, we only focus on symmetric sparsity pattern. Therefore, we assume $\mathbf{\Upsilon_{\rdmmtx}}$ is symmetric in the following.
Here, $\mathbf{\Upsilon_{\rdmmtx}}$ can be easily illustrated by its graph, denoted by $\mathcal{G}(\mathcal{N},\mathcal{E})$, which consists of a set of nodes $\mathcal{N}$ and a set of edges $\mathcal{E}$ between nodes\cite{osterby1983direct}.
In this graph, the node set $\mathcal{N}$ contains $N_{\rdmmtx}$ nodes that represent the corresponding row or column indices of $\mathbf{\rdmmtx}$; the edge set $\mathcal{E}$ includes all the undirected lines connecting from node $i$ to node $j$ if $\mathbf{\rdmmtx}(i,j)$ and $\mathbf{\rdmmtx}(j,i)$ are nonzero.
Besides, the number of edges connected to a node is defined as the degree, denoted by $D$.

On the other hand, we need to take the randomness of $\mathbf{\rdmmtx}$ into consideration.
Since $\mathbf{\rdmmtx}$ is a random sparse matrix, degree of its nodes can be denoted by a discrete random variable ${\varphi}_S$, where $0\leq {\varphi}_S \leq N_{\rdmmtx}-1$.
Based on this, the sparsity level of $\mathbf{\rdmmtx}$ is statistically measured as follows.

\vspace{-0pt}
\begin{definition} \label{def2} 
  (Sparsity level)
  For an arbitrary random sparse matrix $\mathbf{\rdmmtx}$, let $F_{{\varphi}_S}(D)=P({\varphi}_S\leq D)$ denote the CDF of ${\varphi}_S$, where $D$ denotes a specific value of degrees and $D=0,1,\dots,N_{\rdmmtx}-1$.
  This CDF $F_{{\varphi}_S}(D)$ is defined as the sparsity level of $\mathbf{\rdmmtx}$.
\end{definition}
\vspace{-0pt}

With the help of the above definition, we can easily compare the sparsity level of two random sparse matrices.
For two random sparse matrices $\mathbf{\rdmmtx}_1$ and $\mathbf{\rdmmtx}_2$, let $\varphi_{S_1}$ and $\varphi_{S_2}$ denote the node degrees of $\mathbf{\rdmmtx}_1$ and $\mathbf{\rdmmtx}_2$ respectively.
If their sparsity levels satisfy $F_{\varphi_{S_1}}(D)>F_{\varphi_{S_2}}(D)$ for any $D$, more nodes of $\mathbf{\rdmmtx}_1$ have small degrees than those of $\mathbf{\rdmmtx}_2$.
Therefore, $\mathbf{\rdmmtx}_1$ is sparser than $\mathbf{\rdmmtx}_2$.
As will be shown in Section \ref{compl}, the average complexity of our proposed equalizer is largely dictated by the sparsity level of $\mathbf{A}^{-1}$.

Next, based on graph theory, we propose two guidelines to sparsify $\mathbf{\rdmmtx}$ from the edge and node aspects respectively so that the complexity of deriving $\mathbf{\rdmmtx}^{-1}$ can be greatly reduced.
\begin{enumerate}
  \item {\bf Guideline 1} (Sparsify $\mathbf{S}$ from edge constraint): Normalize the diagonal elements of $\mathbf{\rdmmtx}$ through Jacobi scaling as
          $\mathbf{\widetilde{\rdmmtx}}=\mathbf{J}^{-\frac{1}{2}}\mathbf{\rdmmtx}\mathbf{J}^{-\frac{1}{2}}$, where $\mathbf{J}$ is a diagonal matrix and $\mathbf{J}(i,i)=|\mathbf{\rdmmtx}(i,i)|$ for $i=1,\dots,N_{\rdmmtx}$\cite{sedlacek2012sparse}.
        Then, if an edge satisfies $ |\mathbf{\widetilde{\rdmmtx}}(i,j)|\leq \epsilon_{A} $, then this edge is neglected and $\mathbf{{\rdmmtx}}(i,j)$ is set to be zero.
        Here, $\epsilon_{A}$ is an appropriate threshold, and it is preferable to let $\epsilon_{A}\ll 1$ in case of significant performance degradation;
  \item {\bf Guideline 2} (Sparsify $\mathbf{S}$ from node constraint): If the degree of a node is smaller than a threshold $\epsilon_D$, i.e., $D \leq \epsilon_D$, this node can be disconnected from other nodes. 
  To avoid incurring too much imprecision, $\epsilon_D$ should be much smaller than the maximum degree.
\end{enumerate}
The first guideline eliminates edges whose magnitude is negligible compared with the corresponding diagonal element, since these edges have little impact on the matrix operations of $\mathbf{\rdmmtx}$.
Here, we use the Jacobi scaling to normalize the diagonal elements for ease of choosing $\epsilon_A$.
When the diagonal elements of $\mathbf{\rdmmtx}$ have larger magnitude than most of the off-diagonal ones, guideline 1 can effectively sparsify $\mathbf{\rdmmtx}$.
Moreover, the second guideline further enhances the sparsity level  by isolating the nodes with small degrees, since these nodes are hardly related with other nodes.

According to the guidelines, a smaller $\epsilon_{A}$ and a larger $\epsilon_D$ result in a sparser $\mathbf{\rdmmtx}$.
Once $\mathbf{\rdmmtx}$ is well sparsified, the computational complexity of calculating $\mathbf{\rdmmtx}^{-1}$ can be greatly reduced.
However, this complexity reduction is achieved at the cost of performance degradation.
Hence, $\epsilon_{A}$ and $\epsilon_D$ will be carefully chosen from considerable simulations in Section \ref{section7} to strike a balance between the complexity of calculating $\mathbf{\rdmmtx}^{-1}$ and performance degradation.

\vspace{-0pt}
\subsection{Apply the Sparsification Guidelines to our Equalizer}
Now, we apply the above definitions to $\mathbf{A}$, and leverage the sparsification guidelines to reduce the complexity of our equalizer.
From (\ref{HDDp}) and (\ref{Hp}), $\mathbf{H}_{\mathrm{DD}}$ is the sum of $\mathbf{H}_{\mathrm{DD}}^{(p)}$, and $\mathbf{H}_{\mathrm{DD}}^{(p)}$ contains only one nonzero element in each row and column, whose sparsity pattern is given by
\vspace{-0pt}
\begin{equation} \label{spHDDp}
  \mathbf{\Upsilon}_{\mathbf{H}_{\mathrm{DD}}^{(p)}}=
  \mathrm{sp}(\mathbf{H}_{\mathrm{DD}}^{(p)})=\mathbf{\Pi}_N^{k_p}\otimes \mathbf{\Pi}_M^{l_p},
  \vspace{-0pt}
\end{equation}
where $\mathrm{sp}(\cdot)$ denotes the sparsity pattern of a matrix.
Then, the sparsity pattern of $\mathbf{A}$ is
\vspace{-0pt}
\begin{equation} \label{spA}
  \begin{aligned}
    \mathbf{\Upsilon_{A}} & = \mathrm{sp}(\mathbf{H}_{\mathrm{DD}} \mathbf{V} \mathbf{H}_{\mathrm{DD}}^{\mathrm{H}}+N_0\mathbf{I} ) \\&=  
    \mathrm{sp}\left( \sum_{p=1}^P  \mathbf{\Pi}_N^{k_p}\otimes \mathbf{\Pi}_M^{l_p} \cdot \sum_{p^{\prime}=1}^P \mathbf{\Pi}_N^{-k_{p^{\prime}}}\otimes \mathbf{\Pi}_M^{-l_{p^{\prime}}} \right)\\
    &= \mathrm{sp}\left( \sum_{p=1}^P \sum_{p^{\prime}=1}^P \mathbf{\Pi}_N^{k_p-k_{p^{\prime}}}\otimes \mathbf{\Pi}_M^{l_p-l_{p^{\prime}}}   \right)\\
    &=\sum_{(\delta_k,\delta_l)\in\mathcal{D}} (\mathbf{\Pi}_N^{\delta_k}\otimes \mathbf{\Pi}_M^{\delta_l})
  \end{aligned},
  \vspace{-0pt}
\end{equation}
where the identity matrix is dropped, since it is readily proven that the diagonal elements of $\mathbf{H}_{\mathrm{DD}} \mathbf{V} \mathbf{H}_{\mathrm{DD}}^{\mathrm{H}}$ are positive.
Here, $\delta_k=k_p-k_{p^{\prime}}$, $\delta_l=l_p-l_{p^{\prime}}$ for $p,{p^{\prime}}=1,\dots,P$, and $\mathcal{D}$ denotes the set of pairs $(\delta_k,\delta_l)$.
We observe that if $(\delta_k,\delta_l)\in\mathcal{D}$, then $(-\delta_k,-\delta_l)\in\mathcal{D}$, which is consistent with the fact that $\mathbf{\Upsilon_{A}}$ is symmetric.
Besides, the degree of each node $\varphi_A$ is the number of nonzero off-diagonal elements in this column.
According to (\ref{spA}), all the node degrees satisfy $\varphi_A=|\mathcal{D}|-1$.

An example is presented to illustrate the sparsity patterns of $\mathbf{H}_{\mathrm{DD}}$ and $\mathbf{A}$ as shown in \figurename~\ref{spHDDA}.
Assume there are $P=3$ paths with the delay indices $[0,1,2]$ and Doppler indices $[-1,1,1]$, and a small OTFS frame with $M=4$ and $N=3$ is used for convenience of analysis.
In \figurename~\ref{spHDDA} (a), different colors of the square units denote the nonzero elements caused by different paths, i.e., $\mathbf{H}_{\mathrm{DD}}^{(p)}$ for $p=1,2,3$.
In \figurename~\ref{spHDDA} (b), different kinds of the square units denote the nonzero elements caused by different pairs in $\mathcal{D}$, and the square units of pairs $(\delta_k,\delta_l)$ and $(-\delta_k,-\delta_l)$ share the same color but differ in texture.
The graph of sparse matrix $\mathbf{A}$ in \figurename~\ref{spHDDA} is shown in \figurename~\ref{GraphA} (a).
Obviously, the nodes are connected intricately where each node has degree $D=6$.
The resultant $\mathbf{A}^{-1}$ turns out to be a dense matrix whose sparsity pattern is shown in \figurename~\ref{GraphA} (d), and each node of $\mathbf{A}^{-1}$ has degree $D=MN-1$, i.e., $F_{\varphi_{A^{-1}}}(MN-2) = 0$ and $F_{\varphi_{A^{-1}}}(MN-1) = 1$.
This is the worst case that needs to be avoided.

\begin{figure}[t]

  \centering
  \subfigure[Sparsity pattern of $\mathbf{H}_{\mathrm{DD}}$]{
  \begin{minipage}[t]{0.5\linewidth}
    \centering
    \includegraphics[width=1.3in]{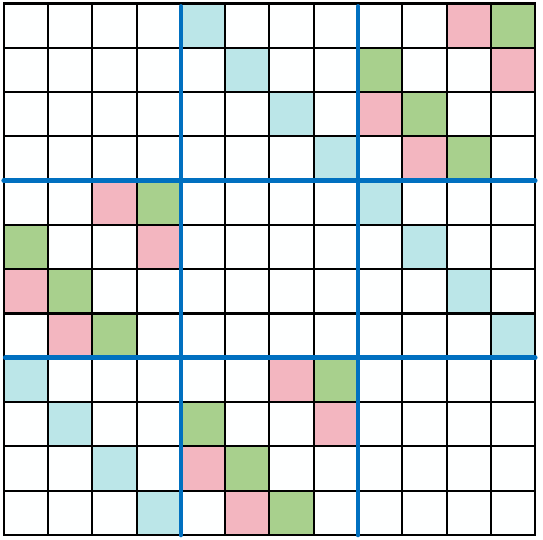}
    %\caption{fig1}
  \end{minipage}%
  }%
  \subfigure[Sparsity pattern of $\mathbf{A}$]{
    \begin{minipage}[t]{0.5\linewidth}
      \centering
      \includegraphics[width=1.3in]{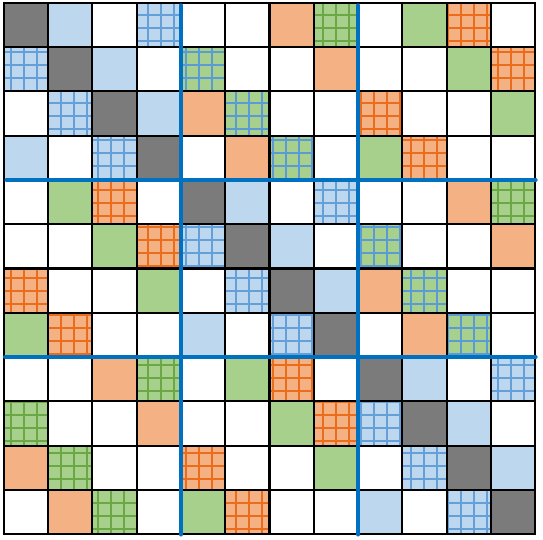}
      %\caption{fig2}
    \end{minipage}%
  }%

  \centering
  \caption{ An example to show the sparsity patterns of $\mathbf{H}_{\mathrm{DD}}$ and $\mathbf{A}$ with $M=4$ and $N=3$, where the delay and Doppler indices are $[0,1,2]$ and $[-1,1,1]$ respectively. } \label{spHDDA}
  \vspace{0pt}
\end{figure}

\begin{figure}[t]

  \centering
  \subfigure[complete graph]{
    \begin{minipage}[t]{0.3\linewidth}
      \centering
      \includegraphics[width=1in]{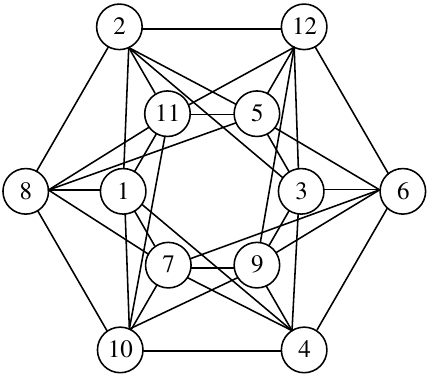}
      %\caption{fig1}
    \end{minipage} \label{Grapha}%
  }%
  \centering
  \subfigure[partially split graph]{
    \begin{minipage}[t]{0.3\linewidth}
      \centering
      \includegraphics[width=1in]{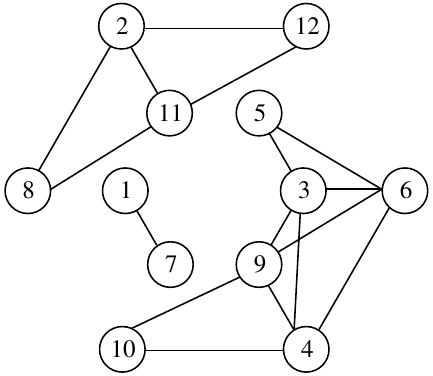}
      %\caption{fig1}
    \end{minipage} \label{Graphb}%
  }%
  \centering
  \subfigure[fully split graph]{
    \begin{minipage}[t]{0.3\linewidth}
      \centering
      \includegraphics[width=1in]{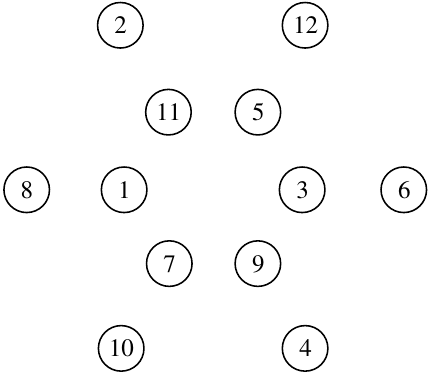}
      %\caption{fig1}
    \end{minipage} \label{Graphc}%
  }%

  \centering
  \subfigure[$\mathbf{\Upsilon}_{\mathbf{A}^{-1}}$ of (a)]{
    \begin{minipage}[t]{0.3\linewidth}
      \centering
      \includegraphics[width=1in]{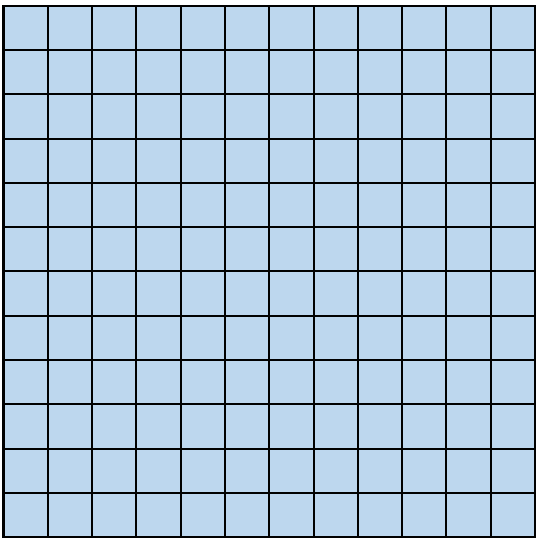}
      %\caption{fig1}
    \end{minipage} \label{Graphd}%
  }%
  \centering
  \subfigure[$\mathbf{\Upsilon}_{\mathbf{A}^{-1}}$ of (b)]{
    \begin{minipage}[t]{0.3\linewidth}
      \centering
      \includegraphics[width=1in]{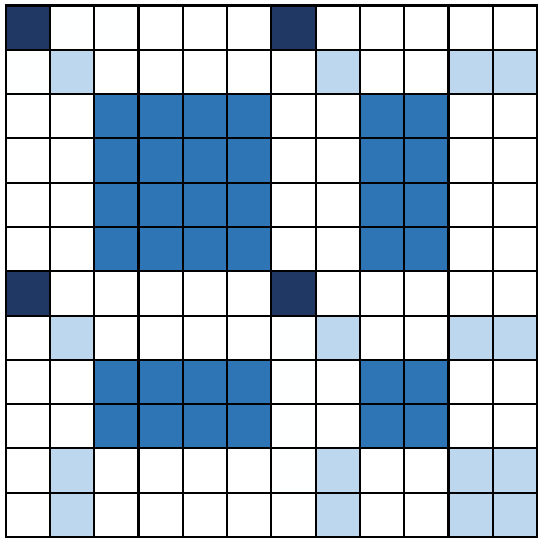}
      %\caption{fig1}
    \end{minipage} \label{Graphe}%
  }%
  \centering
  \subfigure[$\mathbf{\Upsilon}_{\mathbf{A}^{-1}}$ of (c)]{
    \begin{minipage}[t]{0.3\linewidth}
      \centering
      \includegraphics[width=1in]{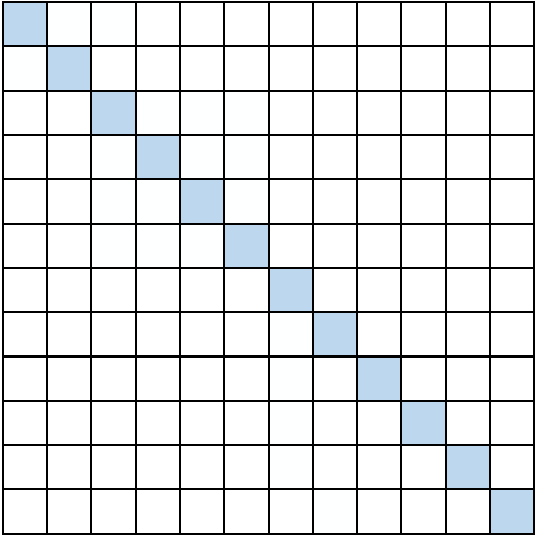}
      %\caption{fig1}
    \end{minipage} \label{Graphf}%
  }%

  \centering
  \caption{Complete and split graphs of $\mathbf{A}$ with their corresponding sparsity patterns of $\mathbf{A}^{-1}$, where the blank square units represent zero elements.}
  \label{GraphA}
  \vspace{-5pt}
\end{figure}

Here, we give a special case to illustrate the effect of the guidelines.
As shown in \figurename~\ref{GraphA} (b), after eliminating all the trivial edges and isolating all the nodes with small degree, the graph is split into three separate subgraphs.
In this way, $\mathbf{A}^{-1}$ can be derived by computing the inverse of these subgraphs individually.
In \figurename~\ref{GraphA} (e), we use different color shades to illustrate the inverse of these subgraphs.
Obviously, the sparsity level of $\mathbf{A}^{-1}$ is enhanced and thus can be derived with less complexity.
As will be shown shortly, even if the graph is not ideally split into several disconnected subgraphs, the complexity can be significantly reduced with the help of the guidelines.
In an extreme case, if all the nodes are fully disconnected, $\mathbf{A}$ is perfectly reduced to a diagonal matrix whose graph is exhibited in \figurename~\ref{GraphA} (c), and the sparsity pattern of $\mathbf{A}^{-1}$ is also diagonal as shown in \figurename~\ref{GraphA} (f).
This is the best case, where $F_{\varphi_{A^{-1}}}(0)=1$ and $\mathbf{A}^{-1}$ can be directly obtained with complexity $\mathcal{O}(MN)$.

As illustrated in the example, $\mathbf{A}$ can be computed with much less complexity after applying these two guidelines.
When $E_b/N_0$ is low, $\mathbf{A}=\mathbf{H}_{\mathrm{DD}}\mathbf{V}\mathbf{H}_{\mathrm{DD}}^{\mathrm{H}}+N_0 \mathbf{I}$ is dominated by the second term $N_0\mathbf{I}$, so that the diagonal elements of $\mathbf{A}$ have much larger magnitude than the off-diagonal elements.
Hence, after sparsifying $\mathbf{A}$ according to the guidelines, most off-diagonal elements can be eliminated.
Unfortunately, at high $E_b/N_0$ the off-diagonal elements are likely to have the same order of magnitude as the diagonal elements.
In this case, we should avoid computing $\mathbf{A}^{-1}$ directly.
Nevertheless, with the help of the decoder, the {\it a priori} information is highly reliable at the second iteration, so that the updated variances are close to zero, which diminish the impact of the first term $\mathbf{H}_{\mathrm{DD}} \mathbf{V} \mathbf{H}_{\mathrm{DD}}^{\mathrm{H}}$ on $\mathbf{A}$.  
As a result, the second term $N_0\mathbf{I}$ dominates again and the guidelines can be utilized to sparsify $\mathbf{A}$.
In summary, reducing the complexity of MMSE estimation can be decomposed into two subproblems.
The first subproblem is to perform MMSE estimation without deriving $\mathbf{A}^{-1}$ directly at the initial iteration when the covariance matrix $\mathbf{A}$ can hardly be sparsified by the guidelines; 
the second one is to derive $\mathbf{A}^{-1}$ at the subsequent iterations when $\mathbf{A}$ has been sparsified according to the proposed guidelines.

\section{Low-Complexity Algorithms for MMSE estimation}
\label{section5}
In this section, we present two low-complexity algorithms to compute $\mathbf{A}^{-1}$, where one of them is used at the initial outer iteration and the other is applied at the subsequent iterations.
The former resorts to the GMRES algorithm to convert the problem of computing $\mathbf{A}^{-1}$ into solving the equivalent sparse linear systems, while the latter leverages the FSPAI algorithm to derive an approximation of $\mathbf{A}^{-1}$ after applying the proposed sparsification guidelines.

\vspace{-0pt}
\subsection{GMRES at the Initial Outer Iteration}
At the initial iteration, we have $\mathbf{A}=\mathbf{H}_{\mathrm{DD}}\mathbf{H}_{\mathrm{DD}}^{\mathrm{H}}+N_0 \mathbf{I}$ due to the lack of {\it a priori} information.
In this case, the off-diagonal elements of $\mathbf{A}$ has a comparable order to the diagonal elements, so that the sparsity level of $\mathbf{A}^{-1}$ can only be slightly enhanced.
Hence, directly deriving $\mathbf{A}^{-1}$ is not a good option.
As described in Section \ref{section3}, $\mathbf{A}^{-1}$ is needed in two places when MMSE estimation is performed, namely $\mathbf{A}^{-1}(\mathbf{y}-\mathbf{H}_{\mathrm{DD}}\mathbf{m})$ of (\ref{MMSELAST}) and $\xi _n= \mathbf{h}_n^{\mathrm{H}} \mathbf{A}^{-1} \mathbf{h}_n$ of (\ref{MMSELAST}) and (\ref{Le}), which are discussed respectively in the following.

With regard to the former, we observe that $\mathbf{A}^{-1}(\mathbf{y}-\mathbf{H}_{\mathrm{DD}}\mathbf{m})$ is a common part of all the estimated symbol $\hat{x}_n$, where $\mathbf{A}^{-1}(\mathbf{y}-\mathbf{H}_{\mathrm{DD}}\mathbf{m})=\mathbf{A}^{-1}\mathbf{y}$ at the initial iteration since $\mathbf{m}$ is of all zeros.
Therefore, we can solve only one equivalent sparse linear system $\mathbf{A}\mathbf{f  }_1={\mathbf{y}}$ for all $\hat{x}_n$'s and avoid directly computing $\mathbf{A}^{-1}$, where $\mathbf{f  }_1$ is the vector of unknowns. 

As for computing $\xi_n$, all $\xi _n$'s for $n=1,\dots, MN$ are almost the same for a given $\mathbf{H}_{\mathrm{DD}}$ as shown in \figurename{}~\ref{v_xi}, where $v_{\xi}=\mathbb{E}\{|\xi_n-\frac{1}{MN}\sum_{n=1}^{MN}\xi_n|^2\}$ is the variance of $\xi_n$.
It is found that $v_{\xi}$ converges to zero when $E_b/N_0$ is extremely low or high.
According to $\xi _n= \mathbf{h}_n^{\mathrm{H}} \mathbf{A}^{-1} \mathbf{h}_n$, $\xi_n$ is actually the diagonal element of $\mathbf{B}=\mathbf{H}_{\mathrm{DD}}^{\mathrm{H}} (\mathbf{H}_{\mathrm{DD}}\mathbf{H}_{\mathrm{DD}}^{\mathrm{H}}+\PSD \mathbf{I})^{-1} \mathbf{H}_{\mathrm{DD}}$.
If $E_b/N_0$ is low enough, $\mathbf{B}$ can be approximated as $\mathbf{B}\approx  N_0^{-1}\mathbf{H}_{\mathrm{DD}}^{\mathrm{H}}\mathbf{H}_{\mathrm{DD}}$ whose diagonal elements are all equal to $N_0^{-1}\sum_{p=1}^P |h_p|^2$ and $v_{\xi}\approx 0$.
If $E_b/N_0$ is sufficiently high, we have $\mathbf{B}\approx\mathbf{H}_{\mathrm{DD}}^{\mathrm{H}} (\mathbf{H}_{\mathrm{DD}}\mathbf{H}_{\mathrm{DD}}^{\mathrm{H}})^{-1} \mathbf{H}_{\mathrm{DD}}=\mathbf{I}$ and thus $v_{\xi}\approx 0$.
Moreover, as observed in \figurename{}~\ref{v_xi}, even the maximum value of $v_{\xi}$ is smaller than $10^{-5}$.
Hence, it is reasonable to assume that $\xi_n$'s are the same for $n=1,\dots,MN$.
Then, deriving $\xi_n$'s can be converted to solving another sparse linear system $\mathbf{A}\mathbf{f  }_2=\mathbf{h}_n$ for an arbitrary $n$, where $\mathbf{f  }_2 = \mathbf{A}^{-1}\mathbf{h}_n$ is the vector of unknowns.
Thereafter, little computation is required to obtain $\xi _n$ through $\xi _n = \mathbf{h}_n^{\mathrm{H}}\mathbf{f  }_2$, since $\mathbf{h}_n$ is sparse with only $P$ nonzero elements.
In this way, instead of directly calculating $\mathbf{A}^{-1}$, we can solve two equivalent sparse linear systems, i.e., $\mathbf{A}\mathbf{f  }_1={\mathbf{y}}$ and $\mathbf{A}\mathbf{f  }_2=\mathbf{h}_n$ at the initial iteration.

\begin{figure}[t]
  \centering
  \includegraphics[width = 0.4\textwidth]{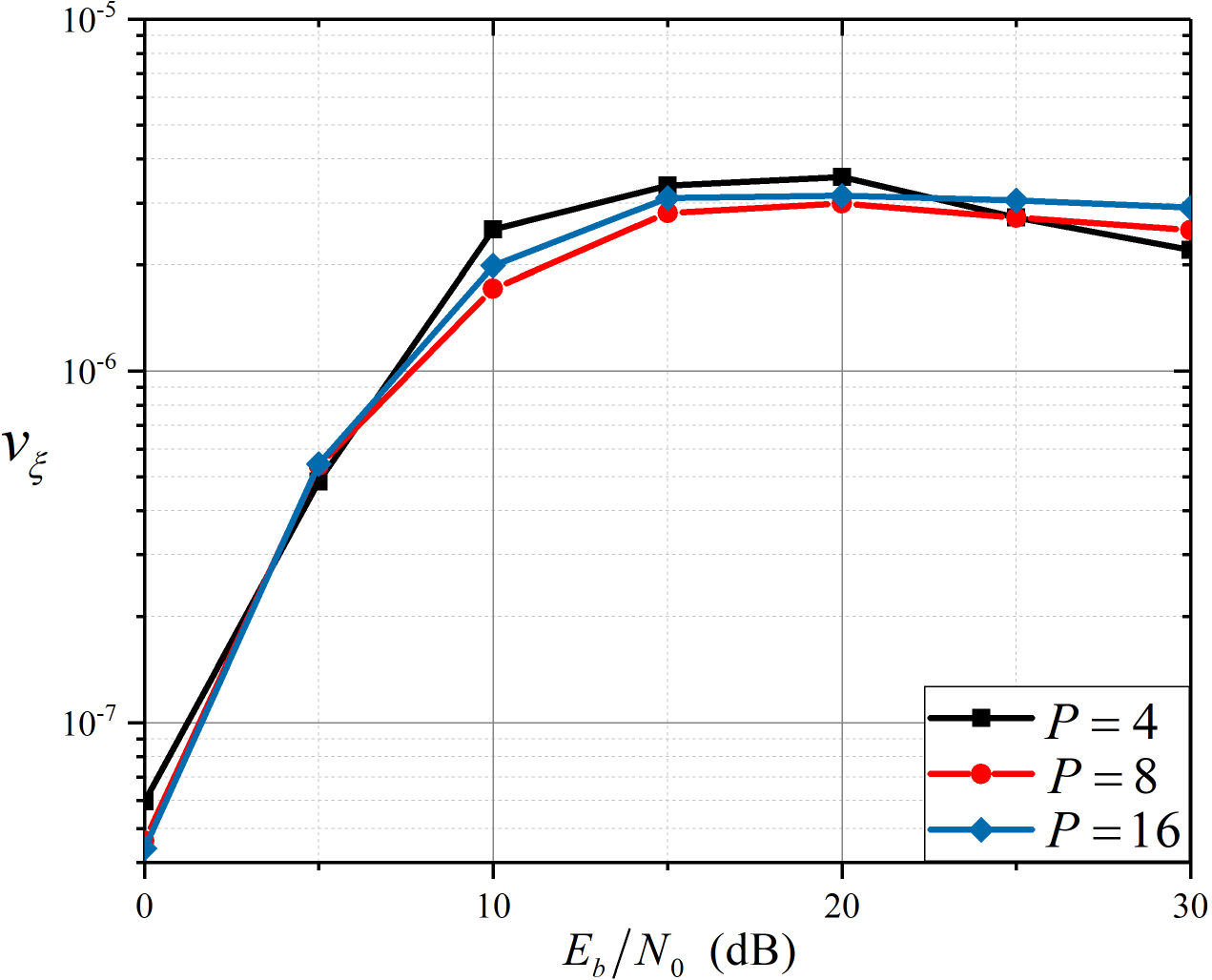}
  \caption{The variance of $\xi_n$ versus $E_b/N_0$, where $M=64$, $N=32$ and $P=4,8,16$. At each $E_b/N_0$ we average over 1000 channel realizations. }\vspace{-1em}
  \label{v_xi}

  \vspace{-0pt}
\end{figure}

We use an iterative algorithm called GMRES to solve the sparse linear systems, since GMRES has no requirements for the coefficient matrix $\mathbf{A}$ of a linear system $\mathbf{A}\mathbf{f  }=\mathbf{b}$.
Besides, GMRES takes less computation than some other candidate iterative algorithms, such as the Generalized Conjugate Residual (GCR) algorithm \cite{saad1986gmres}.
More importantly, the global convergence of GMRES applied in our system can be guaranteed with limited iterations (as will be presented in Section \ref{Conve}).
Note that if MMSE estimation is only performed once without subsequent outer iterations, the iterative equalizer is reduced to a conventional LMMSE equalizer.
Hence, GMRES can be readily extended to any other systems with LMMSE equalizer.
Next, we will briefly introduce how to apply GMRES at the initial iteration.

For a linear system $\mathbf{A}\mathbf{f  }=\mathbf{b}$, GMRES uses an arbitrary initial guess $\mathbf{f  }_0$, e.g., a zero vector, to extract an approximate solution $\mathbf{f  }_j=\mathbf{f  }_0+\Delta\mathbf{f  }_j$ iteratively, where $\Delta\mathbf{f  }_j$ is obtained from a subspace $\mathcal{K}_j\in\mathbb{C}^j$ to minimize the current residual norm $\left\|\mathbf{r}_j\right\| = \left\| \mathbf{b}-\mathbf{A}\mathbf{f  }_j \right\|$ \cite{saad1986gmres}.
Here, $\mathcal{K}_j$ is the $j$-th Krylov subspace
$
  \mathcal{K}_j = \mathit{span}\{ \mathbf{r}_0,\mathbf{Ar}_0,\dots,\mathbf{A}^{j-1}\mathbf{r}_0 \}
$
with the initial residual $\mathbf{r}_0 = \mathbf{b-A}\mathbf{f  }_0$.
For convenience, we set the norm $\left\|\mathbf{r}_0 \right\|=\rho$.
Then the least-squares (LS) problem can be written as
\vspace{-0pt}
\begin{equation} \label{mini}
  \min \limits_{\mathbf{\Delta \mathbf{f  }}_j\in \mathcal{K}_j} \left\| \mathbf{b}-\mathbf{A}(\mathbf{f  }_0+\Delta \mathbf{f  }_j) \right\|= \min \limits_{\mathbf{\Delta \mathbf{f  }}_j\in \mathcal{K}_j} \left\| \mathbf{r}_0-\mathbf{A}\Delta \mathbf{f  }_j \right\|.
  \vspace{-0pt}
\end{equation}
When the residual norm is smaller than a drop tolerance $\epsilon _g$, the algorithm stops.
Otherwise, the subspace needs to be expanded to $\mathcal{K}_{j+1}$, from which the approximate solution is updated to $\mathbf{f  }_{j+1}$ to derive a more precise solution.
This update process is an inner iteration.

However, with the expansion of the subspace, the computational cost increases exponentially.
Hence, we let $\maxGMRES$ be the maximum dimension of the Krylov subspace.
After the $\maxGMRES$-th inner iteration, if the drop tolerance $\epsilon_g$ is still not reached, then compute the current approximate solution $\mathbf{f  }_{\maxGMRES}$ and restart GMRES with a new initial guess $\mathbf{f  }_0^{\prime}=\mathbf{f  }_{\maxGMRES}$.
These $\maxGMRES$ inner iterations is called a restart cycle, and the restarting step can control the overall complexity of GMRES.

Considering that the LS problem of (\ref{mini}) dominates the complexity, we expect to reduce (\ref{mini}) into a smaller problem\cite{arnoldi1951principle}.
We first use the modified Gram-Schmidt algorithm to yield the $\ell_2$-orthonormal basis $\{\mathbf{p}_1,\mathbf{p}_2,\dots,\mathbf{p}_j\}$ of $\mathcal{K}_j$ as summarized in lines 4-8 of Algorithm \ref{GMRESalgo}\cite{saad2003iterative}, where $\mathbf{p}_i\in\mathbb{C}^{MN}$ for $i=1,\dots,j$.
From the modified Gram-Schmidt algorithm, $\mathbf{A}\mathbf{p}_j$ is rewritten as 
\vspace{-0pt}\begin{equation} \label{Ap}
  \mathbf{A}\mathbf{p}_j = \sum_{i=1}^{j+1} t_{i,j}\mathbf{p}_i,
  \vspace{-0pt}
\end{equation}
where $t_{i,j}=(\mathbf{Ap}_j,\mathbf{p}_i)$.
Therefore, we have 
$ \mathbf{A}\mathbf{P}_j  =[\mathbf{A}\mathbf{p}_1,\dots,\mathbf{A}\mathbf{p}_j]= \mathbf{P}_{j+1}\mathbf{{T}}_j
$
where 
% $j=1,2,\dots,\maxGMRES$ and 
$\mathbf{{T}}_j$ is an upper Hessenberg matrix
\vspace{-0pt}
\begin{equation*}
  \setlength{\arraycolsep}{2pt}
  \renewcommand{\arraystretch}{0.6}
  \mathbf{{T}}_j = \left[\begin{array}{ccccc}
    t_{1,1} & \cdots & t_{1,j}   \\
    t_{2,1} & \cdots & t_{2,j}   \\
            & \ddots & \vdots    \\
            &        & t_{j+1,j}
  \end{array}\right]_{(j+1)\times j}.
    \vspace{-0pt}
\end{equation*}
Besides, since $\Delta \mathbf{f  }_j$ is derived from $\mathcal{K}_j$, $\Delta \mathbf{f  }_j$ can be denoted by a linear combination of the orthonormal basis as $\Delta \mathbf{f  }_j = \mathbf{P}_j \mathbf{k}_j$ where $\mathbf{P}_j=[\mathbf{p}_1,\mathbf{p}_2,\dots,\mathbf{p}_j]$ and $\mathbf{k}_j\in\mathbb{C}^{j}$. 
By substituting $\Delta \mathbf{f  }_j = \mathbf{P}_j \mathbf{k}_j$ into (\ref{mini}), the LS problem reduces to
\vspace{-0pt}
\begin{equation}
  \begin{aligned} \label{minJ} 
    \min \limits_{\mathbf{k}_j}\left\| \rho\mathbf{p}_1-\mathbf{A}\mathbf{P}_j \mathbf{k}_j \right\|
    =&  \min \limits_{\mathbf{k}_j} \left\| \rho\mathbf{P}_{j+1}\mathbf{{e}}_1^{(j)}-\mathbf{P}_{j+1}\mathbf{{T}}_j\mathbf{k}_j \right\|
    \\=& \min \limits_{\mathbf{k}_j}\left\| \rho\mathbf{{e}}_1^{(j)}-\mathbf{{T}}_j\mathbf{k}_j \right\|,
  \end{aligned}
  \vspace{-0pt}
\end{equation}
where $\mathbf{{e}}_i^{(j)}$ is the $i$-th column of the identity matrix $\mathbf{I}_{j+1}$.
% Through the Arnoldi's method, the LS problem of (\ref{mini}) is reduced to (\ref{minJ}).
The computation of solving (\ref{minJ}) is trivial since $\mathbf{{T}}_j$ is a $j\times(j+1)$ matrix, where $j$ has a maximum value $\maxGMRES$ and is generally much smaller than $MN$.

% \IncMargin{1em}
% \begin{algorithm}[t]
%   \begin{spacing}{1.1}
%    \SetKwData{Left}{left}\SetKwData{This}{this}\SetKwData{Up}{up} \SetKwFunction{Union}{Union}\SetKwFunction{FindCompress}{FindCompress} \SetKwInOut{Input}{input}\SetKwInOut{Output}{output}

%   {\bf Initialization:} $\mathbf{p}_1=\mathbf{r}_0/\rho$ \;
%   \For{$j=1$ \KwTo $\maxGMRES$}{
%   $\mathbf{\widetilde{p}}_{j+1}= \mathbf{A}\mathbf{p}_j$\;
%   \For{$i=1$ \KwTo $j$}{
%   $t_{i,j} = (\mathbf{A}\mathbf{p}_j,\mathbf{p}_i)$, 
%   $\mathbf{\widetilde{p}}_{j+1}\leftarrow \mathbf{\widetilde{p}}_{j+1}-t_{i,j}\mathbf{p}_i$.
%   }
%   $t_{j+1,j}=\left\| \mathbf{\widetilde{p}}_{j+1} \right\|$,
%   $\mathbf{p}_{j+1}=\mathbf{\widetilde{p}}_{j+1}/t_{j+1,j}$.
%   }
%   \caption{Arnoldi's Method using the Modified Gram-Schmidt}
%   \label{arno}
% \end{spacing}
% \end{algorithm}
% \DecMargin{1em}

In GMRES, whether to start the next iteration of GMRES is determined by the residual norm from (\ref{minJ}).
% Intuitively, to obtain $\|\mathbf{r}_j\|$, $\mathbf{f  }_j$ is necessarily computed through solving (\ref{minJ}).
% If $\|\mathbf{r}_j\|>\epsilon_g$, $\mathbf{f  }_j$ is discarded and another inner iteration is started.  
To avoid solving (\ref{minJ}) repeatedly, we can solve the LS problem of (\ref{minJ}) by transforming $\mathbf{{T}}_j$ into an upper triangular matrix $\mathbf{{U}}_{j}$ column by column using the rotation matrix
\vspace{-0pt}
\begin{equation}
  \setlength{\arraycolsep}{1.6pt}
  \renewcommand{\arraystretch}{0.4}
  \mathbf{\Psi} _{i}^{(j)}=\left[\begin{array}{cccccc}
    1 &           &                    &                   &           &   \\
      & \ddots    &                    &                   &           &      \\
      &           & \varsigma  _{i}    & \vartheta    _{i} &           &   \\
      &           & -\vartheta    _{i} & \varsigma  _{i}   &           &   \\
      &           &                    &                   &    \ddots &   \\
      &           &                    &                   &           & 1
  \end{array}\right]_{(j+1)\times(j+1)},
  \vspace{-0pt}
\end{equation}
where $
\varsigma  _{i} = {t_{i,i}}/ ({t_{i,i}^2+t_{i+1,i}^2})^{1/2}$ and $\vartheta  _i={t_{i+1,i}}/(t_{i,i}^2+t_{i+1,i}^2)^{1/2}$ are in the $i$-th and $(i+1)$-th rows for $i=1,2,\dots,j$.
Using $\mathbf{\mathbf{\Psi}} _{i}^{(j)}$, we derive the upper triangular matrix as
\vspace{-0pt}
\begin{equation} \label{Uj}
  \vspace{-0pt}
  \begin{aligned}
    \setlength{\arraycolsep}{2.2pt}
    \renewcommand{\arraystretch}{0.7}
    \mathbf{{U}}_j=\mathbf{\Psi} _{j}^{(j)}\dots\mathbf{\Psi} _{2}^{(j)}\mathbf{\Psi} _{1}^{(j)}\mathbf{{T}}_j
    =\mathbf{\Xi}_j\mathbf{{T}}_j
     =\left[\begin{array}{ccccc}
              u_{1,1} & \cdots & u_{1,j} \\
                      & \ddots & \vdots  \\
                      &        & u_{j,j} \\
                      &        & 0
                      \end{array}\right],
  \end{aligned}
  \vspace{-0pt}
\end{equation}
where $\mathbf{\Xi}_j =\mathbf{\Psi} _{j}^{(j)}\dots\mathbf{\Psi} _{2}^{(j)}\mathbf{\Psi} _{1}^{(j)}$.
Then, the LS problem is rewritten as
\vspace{-0pt}
\begin{equation}
  \label{wUk}
  \min \limits_{\mathbf{k}_j}\left\| \mathbf{\Xi}_j(\rho \mathbf{{e}}_1^{(j)}-\mathbf{{T}}_j\mathbf{k}_j) \right\|=\min \limits_{\mathbf{k}_j}\left\| \mathbf{w}_j-\mathbf{{U}}_j\mathbf{k}_j \right\|,
  \vspace{-0pt}
\end{equation}
where $\mathbf{w}_j= \rho \mathbf{\Xi}_j\mathbf{{e}}_1^{(j)}$.
Now we can easily solve the LS problem using backward substitution.
Since the last row of $\mathbf{{U}}_j$ is of all zeros, all rows except the last one in $\mathbf{w}_j-\mathbf{{U}}_j\mathbf{k}_j$ can be canceled out by choosing $\mathbf{k}_j$. 
Thus, the current residual is readily derived from the last element of $\mathbf{w}_j$ without any additional computation, i.e., $\|\mathbf{r}_j\|=|\mathbf{w}_j(j+1)|$.
In this way, the LS problem is solved only when the current restart cycle finishes or the algorithm stops.

\IncMargin{1em}
\begin{algorithm}[t]
  \begin{spacing}{1.1}
   \SetKwData{Left}{left}\SetKwData{This}{this}\SetKwData{Up}{up} \SetKwFunction{Union}{Union}\SetKwFunction{FindCompress}{FindCompress} \SetKwInOut{Input}{input}\SetKwInOut{Output}{output}

  {\bf Initialization:} $\mathbf{r}_0=\mathbf{b}-\mathbf{A}\mathbf{f}_0, \rho=\left\| \mathbf{r  }_{0} \right\|, \mathbf{p}_1=\mathbf{r}_0/\rho$, 
  $\mathbf{w}=\rho\mathbf{{e}}_1^{(\maxGMRES)}$, ${\mathbf{T}}$ and $\mathbf{U}$ are $(\maxGMRES+1)\times \maxGMRES$ zero matrix \;
  \For{$j=1$ \KwTo $\maxGMRES$}{
  $\mathbf{\widetilde{p}}_{j+1}= \mathbf{A}\mathbf{p}_j$\;
  \For{$i=1$ \KwTo $j$}{
  $t_{i,j} = (\mathbf{A}\mathbf{p}_j,\mathbf{p}_i)$, 
  $\mathbf{\widetilde{p}}_{j+1}\leftarrow \mathbf{\widetilde{p}}_{j+1}-t_{i,j}\mathbf{p}_i$.
  }
  $t_{j+1,j}=\left\| \mathbf{\widetilde{p}}_{j+1} \right\|$, 
  $\mathbf{p}_{j+1}=\mathbf{\widetilde{p}}_{j+1}/t_{j+1,j}$\;
  $\mathbf{T}(1:j+1,j)\leftarrow[t_{1,j},t_{2,j},\dots,t_{j+1,j}]^{\mathrm{T}} $\;
  Compute $\mathbf{\Psi}_i^{(\maxGMRES)}$, for $i=1,2,\dots,j$\;
  ${\mathbf{U}}(:,j)\leftarrow \mathbf{\Psi}_j^{(\maxGMRES)} \dots \mathbf{\Psi}_2^{(\maxGMRES)} \mathbf{\Psi}_1^{(\maxGMRES)}\mathbf{T}(:,j)$,
  $\mathbf{w}\leftarrow \mathbf{\Psi}_j^{(\maxGMRES)}\mathbf{w}$\;
  \If{$|\mathbf{w}(j+1)|<\epsilon _g $}{
  {\bf Solve:} the LS problem $\min_{\mathbf{k}_j}\left\| \mathbf{w}(1:j+1)-\mathbf{{U}}(1:j+1,1:j)\mathbf{k}_j \right\|$, $\mathbf{f  }_j=\mathbf{f  }_0+\mathbf{P}_j \mathbf{k}_j$\;
  {\bf Stop}
  }
  }
  {\bf Solve:} the LS problem $\min_{\mathbf{k}_{\maxGMRES}}\left\| \mathbf{w}-\mathbf{{U}}\mathbf{k}_{\maxGMRES} \right\|$, $\mathbf{f  }_{\maxGMRES}=\mathbf{f  }_0+\mathbf{P}_{\maxGMRES} \mathbf{k}_{\maxGMRES}$\;
  {\bf Restart:} $\mathbf{r}_0=\mathbf{b}-\mathbf{A}\mathbf{f}_{\maxGMRES}, \rho=\left\| \mathbf{r  }_{0} \right\|$, $\mathbf{p}_1=\mathbf{r  }_{0}/\rho$,
$\mathbf{w}=\rho\mathbf{{e}}_1^{(\maxGMRES)}$, ${\mathbf{T}}$ and $\mathbf{U}$ are $(\maxGMRES+1)\times \maxGMRES$ zero matrix\;
  {\bf Go to Line 3}.
  \caption{Restarted GMRES}
  \label{GMRESalgo}
  \end{spacing}
\end{algorithm}
\DecMargin{1em}

The complete algorithm is summarized in Algorithm \ref{GMRESalgo}. 
Initially, we set $\mathbf{{U}}$ as a $(\maxGMRES+1)\times \maxGMRES$ zero matrix, and transform $\mathbf{{T}}_{\maxGMRES}$ to $\mathbf{{U}}_{\maxGMRES}$ progressively in $\maxGMRES$ inner iterations.
At the $j$-th inner iteration ($j=1,\dots,\maxGMRES$), we use the modified Gram-Schmidt algorithm and obtain $\{t_{1,j},t_{2,j},\dots,t_{j+1,j}\}$ as shown in lines 4-8.
The new column $[t_{1,j},t_{2,j},\dots,t_{j+1,j}]^{\mathrm{T}}$ is first multiplied by all the previous  $\mathbf{\Psi} _{i}^{(\maxGMRES)}$ for $i=1,2,\dots,j-1$, since the previous rotation matrices also affect this newly produced column.
Then, $\mathbf{\Psi} _{j}^{(\maxGMRES)}$ is used to eliminate $t_{j+1,j}$.
Meanwhile, $\mathbf{w}$ is multiplied by $\mathbf{\Psi} _{j}^{(\maxGMRES)}$ according to (\ref{wUk}).
% Additionally, the same operations are applied to $\rho \mathbf{{e}}_1^{(\maxGMRES)}$.
These steps are shown in lines 10-11.
If the current residual $|\mathbf{w}_j(j+1)|>\epsilon _g$, another iteration is performed.
At the end of a restart cycle, if norm $\|\mathbf{r}_{\maxGMRES}\|$ is still larger than $\epsilon _g$, then restart the algorithm and go to line 3.

\vspace{-0pt}
\subsection{FSPAI at the Subsequent Outer Iterations}
After the first outer iteration, we have obtained new means and variances, which serve as the {\it a priori} information to start the next outer iteration.
However, since the variances matrix $\mathbf{V}$ is no longer an identity matrix, $\xi _n$'s have completely different values for $n=1,\dots,MN$.
As a result, GMRES needs to be used $MN$ times per outer iteration to derive each $\xi _n$, which is not acceptable for the subsequent outer iterations.
However, with the help of the decoder, most {\it a prior} information of the estimator is reliable enough, so that the corresponding new variances of the estimated symbols approach zero.
Hence, the second term $\PSD\mathbf{I}$ of $\mathbf{A}$ dominates, and the diagonal elements have much larger magnitude than other elements.
By sparsifying $\mathbf{A}$ according to the guidelines, $\mathbf{A}$ becomes much sparser, and the sparsity level of $\mathbf{A}^{-1}$ can be greatly enhanced. 
Thus, we can use the FSPAI algorithm to yield a sparse approximation of $\mathbf{A}^{-1}$, which is inherently parallelizable and can further save the computational time.

We observe that $\mathbf{A}$ is Hermitian positive definite (HPD), since $\mathbf{A}=\mathbf{H}_{\mathrm{DD}} \mathbf{V} \mathbf{H}_{\mathrm{DD}}^{\mathrm{H}}+\PSD\mathbf{I}$ is the sum of a Gram matrix and an identity matrix\cite{horn2012matrix}.
For an HPD matrix $\mathbf{A}=\mathbf{L}_A^{\mathrm{H}}\mathbf{L}_A$ with Cholesky factor $\mathbf{L}_A$, the inverse $\mathbf{A}^{-1}$ can be obtained by deriving a sparse approximate Cholesky factor $\mathbf{L}$ such that $\mathbf{L}\mathbf{L}^{\mathrm{H}}\approx\mathbf{A}^{-1}$.
Here, $\mathbf{L}$ is a lower triangular matrix and can be obtained by minimizing the Frobenius norm $\|\mathbf{L}_A\mathbf{L}-\mathbf{I}\|_{F}$ with respect to a prescribed sparsity pattern of $\mathbf{L}$, denoted by $\mathbf{\Upsilon_{L}}=\mathcal{L}$ \cite{kolotilina1993factorized}.
This is equivalent to minimizing the Kaporin condition number of $\mathbf{L^{\mathbf{H}} \mathbf{A} \mathbf{L} }$, which is defined as \cite{kaporin1994new}
\vspace{-0pt}
\begin{equation}
  K=
  \frac{ \mathrm{tr} (\mathbf{L}^{\mathrm{H}}  \mathbf{A} \mathbf{L})    }
  { MN \cdot\mathrm{det} ( \mathbf{L}^{\mathrm{H}} \mathbf{A} \mathbf{L} )^{\frac{1}{MN} } },
  \vspace{-0pt}
\end{equation}
where $\mathrm{tr(\cdot)}$ and $\mathrm{det(\cdot)}$ denote the trace and determinant of a matrix respectively, and $K$ measures the quality of the sparse approximation where $K\geq 1$.
Evidently, a more precise $\mathbf{L}$ makes $\mathbf{L}^{\mathrm{H}}  \mathbf{A} \mathbf{L}$ closer to $\mathbf{I}$, and thus $K$ is closer to one.
By minimizing $K$, FSPAI can capture the sparsity pattern of $\mathbf{L}$ dynamically \cite{huckle2003factorized}.
Since the consequent $\mathbf{L}$ is typically sparse, there is a significant reduction in complexity.
As opposed to the conventional FSPAI used in symmetric positive definite linear system, FSPAI applied in our system should be modified to match the HPD case.

\vspace{-3pt}
\begin{figure*}[t]
  \centering
\begin{equation} \label{K}
  \begin{aligned}
    K & =  \frac{ \mathrm{tr} (\mathbf{L}^{\mathrm{H}}  \mathbf{A} \mathbf{L})    }
    { MN \cdot\mathrm{det} ( \mathbf{L}^{\mathrm{H}} \mathbf{A} \mathbf{L} )^{\frac{1}{MN} } }
    = \frac{ \sum_{\fidx=1}^{MN} \mathbf{l}_{\fidx}^{\mathrm{H}} \mathbf{A} \mathbf{l}_{\fidx} }
    { MN\cdot \left[ \mathrm{det}(\mathbf{L}^{\mathrm{H}})
    \mathrm{det}(\mathbf{A}) \mathrm{det}(\mathbf{L}) \right]   ^{\frac{1}{MN} }  } \\
      & = \frac{ \sum_{\fidx=1}^{MN} \left[
    l_{\fidx\fidx}^2 a_{\fidx\fidx} +
    2l_{\fidx\fidx}\mathrm{Re}\{ \mathbf{l}_{\fidx}(\widetilde{\setj}_{\fidx})^{\mathrm{H}} \mathbf{a}_{\fidx}(\widetilde{\setj}_{\fidx})\} +
    \mathbf{l}_{\fidx}(\widetilde{\setj}_{\fidx})^{\mathrm{H}}
    \mathbf{A}(\widetilde{\setj}_{\fidx}, \widetilde{\setj}_{\fidx})
    \mathbf{l}_{\fidx}(\widetilde{\setj}_{\fidx})
    \right] }
    { MN\cdot\mathrm{det}(\mathbf{A})^{\frac{1}{MN} } \left( \prod_{\fidx=1}^{MN} l_{\fidx\fidx} ^{\frac{2}{MN} } \right) }
  \end{aligned},
  \vspace{-0pt}
\end{equation}
\noindent\rule[0.25\baselineskip]{\textwidth}{0.3pt}
\end{figure*}

\vspace{3pt}
Initially, a prescribed sparsity pattern $\mathbf{\Upsilon_{L}}$ is selected, e.g., $\mathbf{\Upsilon_{L}}=\mathbf{I}$ with only diagonal elements nonzero.
Let $\mathbf{a}_{\fidx}$ and $\mathbf{l}_{\fidx}$ be the $\fidx$-th column of $\mathbf{A}$ and $\mathbf{L}$ respectively.
Let $\setj_{\fidx}$ denote the set of nonzero element indices of $\mathbf{l}_{\fidx}$ and $\widetilde{\setj}_{\fidx} = \setj_{\fidx}\backslash \{\fidx\}$.
Then we have (\ref{K}) as shown at the top of the next page,
where $l_{\fidx\fidx}=\mathbf{L}(\fidx,\fidx)\in\mathbb{R}^{+}$ and $a_{\fidx\fidx}=\mathbf{A}(\fidx,\fidx)\in\mathbb{R}^{+}$.
For each column, $l_{\fidx\fidx}$ and $\mathbf{l}_{\fidx}(\widetilde{\setj}_{\fidx})$ are obtained to minimize $K$.
To do so, we let $\partial K /\partial l_{\fidx\fidx}=0$ and $\partial K /\partial \mathbf{l}_{\fidx}(\widetilde{\setj}_{\fidx})=0$, and then
\vspace{-0pt}
\begin{subequations}\label{l}
  \begin{align}
    \label{l1} & \mathbf{q}_{\fidx} = \mathbf{A}(\widetilde{\setj}_{\fidx}, \widetilde{\setj}_{\fidx})^{-1}
    \mathbf{a}_{\fidx}(\widetilde{\setj}_{\fidx}),                                                                                                               \\
    \label{l2} & l_{\fidx\fidx} = \left( a_{\fidx\fidx} - \mathbf{a}_{\fidx}(\widetilde{\setj}_{\fidx})^{\mathrm{H}} \mathbf{q}_{\fidx}  \right)^{-\frac{1}{2}}, \\
    \label{l3} & \mathbf{l}_{\fidx}(\widetilde{\setj}_{\fidx}) = -l_{\fidx\fidx}\mathbf{q}_{\fidx}.
  \end{align}
  \vspace{-15pt}
\end{subequations}

\noindent Since $l_{\fidx\fidx}$ and $\mathbf{l}_{\fidx}(\widetilde{\setj}_{\fidx})$ are independently computed for each column, $\mathbf{L}$ can be obtained column by column in parallel.
Due to the parallelism, we can only consider $\mathbf{l}_{\fidx}$ in the following.

FSPAI iteratively improves the approximation by carefully choosing new index $\newidx$ and adding it to $\setj_{\fidx}$.
This process of augmenting $\setj_{\fidx}$ is called an inner iteration of FSPAI.
The improved approximation is derived as $\mathbf{l}_{\fidx}^{\prime} = \mathbf{l}_{\fidx}+l_{\newidx\fidx}\mathbf{e}_{\newidx}$, where $\mathbf{e}_{\newidx}$ denotes the $\newidx$-th column of $\mathbf{I}$ and $\newidx>\fidx$.
To minimize the new Kaporin condition number $K^{\prime}_{\newidx}$ for index $\newidx$, we need to solve \cite{sedlacek2012sparse}
\vspace{-0pt}
\begin{equation}
  \begin{aligned}
    &\min \limits_{l_{\newidx\fidx}} K^{\prime}_{\newidx}
    \\= &\min \limits_{l_{\newidx\fidx}}
    \frac{ \mathrm{tr} \left(
    (\mathbf{L} + l_{\newidx\fidx} \mathbf{e}_{\newidx} \mathbf{e}_{\fidx}^{\mathrm{H}} )^{\mathrm{H}}
    \mathbf{A}
    (\mathbf{L} + l_{\newidx\fidx} \mathbf{e}_{\newidx} \mathbf{e}_{\fidx}^{\mathrm{H}} )
    \right)    }
    { MN \cdot\mathrm{det} \left(
    (\mathbf{L}^{\mathrm{H}} + l_{\newidx\fidx}^{*} \mathbf{e}_{\fidx} \mathbf{e}_{\newidx}^{\mathrm{H}} )
    \mathbf{A}
    (\mathbf{L} + l_{\newidx\fidx} \mathbf{e}_{\newidx} \mathbf{e}_{\fidx}^{\mathrm{H}} )
    \right)^{\frac{1}{MN} } }                                                                                               \\
    = & 
    \min \limits_{l_{\newidx\fidx}}
    \frac{
    \mathrm{tr}( \mathbf{L}^{\mathrm{H}} \mathbf{A} \mathbf{L} )+
    2\mathrm{Re} \{ l_{\newidx\fidx}^{*}\mathbf{a}_{\newidx}^{\mathrm{H}}\mathbf{l}_{\fidx} \}+
    a_{\newidx\newidx}| l_{\newidx\fidx} |^2
    }{
    MN\cdot\mathrm{det}(\mathbf{A})^{\frac{1}{MN} } \left( \prod_{\fidx=1}^{MN} l_{\fidx\fidx} ^{\frac{2}{MN} } \right)
    }
  \end{aligned},
  \vspace{-0pt}
\end{equation}
which can be derived by letting the derivative $\partial K^{\prime}_{\newidx}/\partial l_{\newidx\fidx}=0$.
Then, we have
\vspace{-0pt}
\begin{equation*}
  l_{\newidx\fidx} = - \frac{\mathbf{a}_{\newidx}^{\mathrm{H}} \mathbf{l}_{\fidx}}{a_{\newidx\newidx}}=
  -\frac{\mathbf{a}_{\newidx}(\setj_{\fidx})^{\mathrm{H}} \mathbf{l}_{\fidx}(\setj_{\fidx})} {a_{\newidx\newidx}},
  \vspace{-0pt}
\end{equation*}
where the sparsity of $\mathbf{l}_{\fidx}$ is taken into account, and thus trivial computation is entailed to derive $l_{\newidx\fidx}$. 
With $l_{\newidx\fidx}$, the minimum of $K^{\prime}_{\newidx}$ for index $\newidx$ is
\vspace{-0pt}
\begin{equation} \label{K2}
  \min \limits_{l_{\newidx\fidx}} K_{\newidx}^{\prime} =
  \left(
  1-\frac{\eta_{\newidx\fidx} }{MN} K
  \right),
  \vspace{-0pt}
\end{equation}
where $\eta_{\newidx\fidx}=|\mathbf{a}_{\newidx}^{\mathrm{H}} \mathbf{l}_{\fidx}|^2/a_{\newidx\newidx}$ measures the reduction in Kaporin condition number by adding the new index $\newidx$.
Due to the sparsity of $\mathbf{A}$, only a small fraction of indices $\newidx$ satisfy $\mathbf{a}_{\newidx}(\setj_{\fidx})^{\mathrm{H}} \mathbf{l}_{\fidx}(\setj_{\fidx})\neq 0$, such that the approximation can be improved as shown in (\ref{K2}).
Resorting to the graph of $\mathbf{A}$ in Section \ref{section4}, 
we can readily find these indices that are defined as
$
  \widehat{\setj}_{\fidx}=\{
  \newidx :\newidx>\fidx ~\text{and}~ \mathbf{a}_{\newidx}(\setj_{\fidx})^{\mathrm{H}} \mathbf{l}_{\fidx}(\setj_{\fidx})\neq 0
  \} \backslash \setj_{\fidx}.
$
Thereafter, we can add one or more indices from $\widehat{\setj}_{\fidx}$ with the largest $\eta_{\newidx\fidx}$ to $\setj_{\fidx}$, and obtain augmented indices sets $\setj_{\fidx}^{\prime}$ and $\widetilde{\setj}_{\fidx}^{\prime}=\setj_{\fidx}^{\prime}\backslash\{\fidx\}$.
Then a more precise solution $\mathbf{l}_{\fidx}^{\prime}$ is derived according to (\ref{l}).
Note that the complexity of FSPAI is dominated by computing the inverse in (\ref{l1}).
Hence, the cardinality of $\widetilde{\setj}_{\fidx}$ should not exceed a certain value $\zeta$.
The algorithm keeps iterating until $|\widetilde{\setj}_{\fidx}|>\zeta$ or the maximum of $\eta_{\newidx\fidx}$ is smaller than a well-chosen tolerance $\epsilon_f$.
The values of $\zeta$ and $\epsilon_f$ will be optimized from considerable simulations in Section \ref{section7} to reach a balance between the BER performance and complexity.
Especially, $\zeta$ should be carefully chosen, since it largely determines the complexity order of our equalizer.

The entire FSPAI is summarized in Algorithm \ref{FSPAIalgo}, where the initial sparsity pattern of $\mathbf{L}$ is set to be diagonal, and for each column, only one single index with the largest $\zeta_{\newidx\fidx}$ is added to $\setj_{\fidx}$ during each index update.
Other initial sparsity patterns and index update strategies can be readily applied with trivial modifications.

\IncMargin{1em}
\begin{algorithm} [t]
  \begin{spacing}{1.1}
  \SetKwData{Left}{left}\SetKwData{This}{this}\SetKwData{Up}{up} \SetKwFunction{Union}{Union}\SetKwFunction{FindCompress}{FindCompress} \SetKwInOut{Input}{input}\SetKwInOut{Output}{output}

  {\bf Initialization:} $\mathbf{\Upsilon_{L}}=\mathbf{I}$, $\epsilon_f\geq 0$, $\zeta\geq 0$ \;
  \For{$\fidx=1$ \KwTo $MN$}{
  $\setj_{\fidx}\leftarrow \{\fidx\}$, $\widetilde{\setj}_{\fidx}\leftarrow \emptyset$, $l_{\fidx\fidx}\leftarrow\sqrt{a_{\fidx\fidx}}$;\\
  \While{$|\widetilde{\setj}_{\fidx}|< \zeta$}{
  $\widehat{\setj}_{\fidx}=\{
    \newidx :\newidx>\fidx ~\text{and}~ \mathbf{a}_{\newidx}(\setj_{\fidx})^{\mathrm{H}} \mathbf{l}_{\newidx}(\setj_{\fidx})\neq 0
    \} \backslash \setj_{\fidx}$;\\
  $\eta_{\newidx\fidx}\leftarrow| \mathbf{a}_{\newidx}(\setj_{\fidx})^{\mathrm{H}} \mathbf{l}_{\newidx}(\setj_{\fidx}) |^2/a_{\newidx\newidx}$, for each $\newidx\in\widehat{\setj}_{\fidx}$;\\
  $\eta_{\mathrm{max}}=\max_{\newidx\in\widehat{\setj}_{\fidx}} \eta_{\newidx\fidx}$ ;\\
  \If{$\eta_{\mathrm{max}}<\epsilon_f$}{
    {\bf Stop.}
  }
  $\setj_{\fidx}\leftarrow\setj_{\fidx}\cup \{\newidx : \eta_{rk}=\eta_{\mathrm{max}} \}$, $\widetilde{\setj}_{\fidx}\leftarrow {\setj}_{\fidx}\backslash\{\fidx\}$;\\
  $\mathbf{q}_{\fidx} \leftarrow \mathbf{A}(\widetilde{\setj}_{\fidx}, \widetilde{\setj}_{\fidx})^{-1}
    \mathbf{a}_{\fidx}(\widetilde{\setj}_{\fidx})$;\\
  $l_{\fidx\fidx} \leftarrow ( a_{\fidx\fidx} - \mathbf{a}_{\fidx}(\widetilde{\setj}_{\fidx})^{\mathrm{H}} \mathbf{q}_{\fidx}  )^{-\frac{1}{2}}$; \\
  $\mathbf{l}_{\fidx}(\widetilde{\setj}_{\fidx}) \leftarrow -l_{\fidx\fidx}\mathbf{q}_{\fidx}$.
  }
  }
  \caption{FSPAI}
  \label{FSPAIalgo}
  \end{spacing}
\end{algorithm}
\DecMargin{1em}
\vspace{-0pt}

\section{Theoretical Analysis}\label{section6}
\vspace{-0pt}
\subsection{Convergence Analysis}  \label{Conve}
In what follows, full GMRES (without restarting) used in the MMSE estimator is proven to globally converge, and the upper bound of its convergence rate is derived.
Then we prove that for restarted GMRES, $\maxGMRES$ can be chosen arbitrarily without losing the global convergence.

For full GMRES, considering the definition of Krylov subspace, the residual norm at the $j$-th inner iteration for $j\geq 1$ can be written as 
\vspace{-0pt}
\begin{equation} 
  \|\mathbf{r}_j\|=\min \limits_{\mathbf{\Delta \mathbf{f  }}_j\in \mathcal{K}_j} \left\| \mathbf{r}_0-\mathbf{A}\Delta \mathbf{f  }_j \right\|=\min\limits_{p_j\in \mathcal{P} _j} \|p_j(\mathbf{A})\mathbf{r}_0\|,
  \vspace{-0pt}
\end{equation}
where $\mathcal{P} _j$ denotes the set of polynomials satisfying $p_j(\mathbf{A})=\sum_{i=0}^{j}\alpha_i\mathbf{A}^{i}$ and $\alpha_0=1$\cite{liesen2004convergence}.
Since $\mathbf{A}$ is Hermitian, $\mathbf{A}$ can be diagonalized as $\mathbf{A} = \mathbf{U}_A \mathbf{\Lambda} \mathbf{U}_A^{\mathrm{H}}$ with unitary matrix $\mathbf{U}_A$ and diagonal matrix $\mathbf{\Lambda}$.
As proven in \cite{saad1986gmres}, for a diagonalizable $\mathbf{A}$, the relative residual norm satisfies 
\vspace{-0pt}
\begin{equation} \label{ub1}
  \frac{\|\mathbf{r}_j\|}{\|\mathbf{r}_0\|} \leq \min\limits_{p_j\in \mathcal{P} _j} \max\limits_{\lambda_A\in\mathrm{eig}(\mathbf{A})}|p_j(\lambda_A)|,
  \vspace{-0pt}
\end{equation}
where $\mathrm{eig}(\mathbf{A})$ denotes the spectrum of $\mathbf{A}$ and $\lambda_A$ represents the eigenvalues of $\mathbf{A}$.
Since $\mathbf{A}$ is HPD, we have $\lambda_A\in\mathbb{R}^{+}$ and (\ref{ub1}) can be bounded as \cite{saad1981krylov}
\vspace{-0pt}
\begin{equation} \label{ConvRate}
   \frac{\|\mathbf{r}_j\|} {\|\mathbf{r}_0\|} \leq \min\limits_{p_j\in \mathcal{P} _j} \max\limits_{\lambda\in
    [\lambda_{\mathrm{min}},\lambda_{\mathrm{max}}]
  } |p_j(\lambda)|
  =[ T_j(\varrho ) ]^{-1},
  \vspace{-0pt}
\end{equation}
where $[\lambda_{\mathrm{min}},\lambda_{\mathrm{max}}]$ is an interval in $\mathbb{R}$ with $\lambda_{\mathrm{min}}$ and $\lambda_{\mathrm{max}}$ denoting the minimal and maximal eigenvalues of $\mathbf{A}$ respectively, and $\varrho =  ( \lambda_{\mathrm{max}} + \lambda_{\mathrm{min}} ) / ( \lambda_{\mathrm{max}}-\lambda_{\mathrm{min}} )$ satisfying $\varrho\geq 1$.
Here, $T_j$ is the Chebyshev polynomial of degree $j$ of the first kind as
\vspace{-0pt}
\begin{equation} \label{Tj}
  T_j(\varrho )=\frac{1}{2}\left[
    \left( \varrho+\sqrt{\varrho^2-1} \right)^j +
    \left( \varrho-\sqrt{\varrho^2-1} \right)^j
    \right].
    \vspace{-0pt}
\end{equation}
If $\lambda_{\mathrm{max}} = \lambda_{\mathrm{min}}\equiv \lambda$, i.e., $\varrho = 1$, then $\mathbf{A}=\lambda\mathbf{U}_A \mathbf{I} \mathbf{U}_A^{\mathrm{H}}=\lambda\mathbf{I}$.
This happens if and only if $\mathbf{H}_{\mathrm{DD}}$ is diagonal, i.e., $P=1$ with $\tau=0$ and $\nu=0$.
This indicates that there is only one time-invariant path, which conflicts with the assumption of the doubly selective channel, and thus $\varrho> 1$.

\vspace{-0pt}
\begin{thm} \label{thm1} 
  Full GMRES applied to MMSE estimation globally converges with a convergence rate at least $[T_j(\varrho)]^{-1}$.
\end{thm}
\begin{proof}
  Applying Jessen's inequality to the concave function $x^n$ yields the inequality $(x_1^n+x_2^n)/2\geq [(x_1+x_2)/2]^n$.
  Then the Chebyshev polynomial in (\ref{Tj}) satisfies
  \vspace{-0pt}
  \begin{equation}
    \begin{aligned}
    [T_j(\varrho )]^{-1} \leq& \left( 
    {\left( \varrho+\sqrt{\varrho^2-1}  +  \varrho-\sqrt{\varrho^2-1} \right)/2} 
    \right)^{-j}\\=&\varrho^{-j}< 1,
    \vspace{-0pt}
    \end{aligned}
  \end{equation}
  for any $j\geq 1$.
  Thus, according to (\ref{ConvRate}), the relative residual norm decreases monotonically with $j$, and thus the GMRES algorithm converges at least as rapidly as $[T_j(\varrho)]^{-1}$.
\end{proof}

From Theorem \ref{thm1}, we can infer the relation between $E_b/N_0$ and the convergence rate as follows.

\vspace{0pt}
\begin{corollary} \label{prop1}
  The convergence rate of full GMRES decreases with $E_b/N_0$.
\end{corollary}
\begin{proof}
  For $\mathbf{A}=\mathbf{H}_{\mathrm{DD}}\mathbf{V}\mathbf{H}_{\mathrm{DD}}^{\mathrm{H}}+N_0\mathbf{I}$, we have $\lambda_A=\overline{\lambda}+N_0$, where $\overline{\lambda}$ denotes the eigenvalues of $\mathbf{H}_{\mathrm{DD}}\mathbf{V}\mathbf{H}_{\mathrm{DD}}^{\mathrm{H}}$.
  Then $\varrho$ can be rewritten as
  \vspace{-0pt}
  \begin{equation}
    \varrho=\frac{ {\lambda}_{\mathrm{max}} + \lambda_{\mathrm{min}} }{ \lambda_{\mathrm{max}} - \lambda_{\mathrm{min}} }
    =\frac{ {\overline{\lambda}}_{\mathrm{max}} + {\overline{\lambda}}_{\mathrm{min}} + 2\PSD }{ {\overline{\lambda}}_{\mathrm{max}} - {\overline{\lambda}}_{\mathrm{min}} }.
    \vspace{-0pt}
  \end{equation}
  Thus, for a specific channel realization $\mathbf{H}_{\mathrm{DD}}$, $\varrho$ is totally governed by $\PSD$.
  Besides, it can be easily proven that $T_j(\varrho)$ monotonically increases with $\varrho$ for $\varrho>1$.
  Since the energy of transmitted symbols is normalized, the increase of $E_b/N_0$ is equivalent to the decrease of $\PSD$.
  Accordingly, $\varrho$ and $T_j(\varrho)$ decrease.
  Hence, the convergence rate of full GMRES decreases with $E_b/N_0$.
\end{proof}

Next, we discuss the convergence of restarted GMRES, which is slower than that of full GMRES \cite{baker2005technique}.
This is due to the fact that once the algorithm is restarted, the current approximation subspace is discarded, and the residual generated later will not be orthogonal to the previously obtained subspace.
Nevertheless, the convergence of restarted GMRES used in the MMSE estimator can still be guaranteed as stated in Theorem \ref{thm2}.

\vspace{-0pt}
\begin{thm} \label{thm2}
  Restarted GMRES applied to MMSE estimation converges for any $\maxGMRES\geq 1$.
\end{thm}
\begin{proof}
  The proof is given in the Appendix.
\end{proof}

\vspace{-10pt}
\subsection{Complexity Analysis} \label{compl}
In this subsection, we focus on the complexity of GMRES and FSPAI while neglecting other operations that take trivial computational time due to the inherent sparsity, such as $\mathbf{L}\mathbf{L}^{\mathrm{H}}$ and $\xi_n $.
We can infer from (\ref{spA}) that the number of nonzero elements in each column of $\mathbf{A}$ is at most $\psi    =2\tbinom{P}{2}+1$.
This is because for $p={p^{\prime}}$, $(\delta_k,\delta_l)=(0,0)$ which corresponds to the diagonal elements; for $p\neq {p^{\prime}}$, there are at most $\tbinom{P}{2}$ combinations of $(p,{p^{\prime}})$, where each combination has two distinct pairs $(\delta_k,\delta_l)$ and $(-\delta_k,-\delta_l)$ corresponding to the off-diagonal elements.
Moreover, after excluding the diagonal element, we observe that the degree of node satisfies $\varphi_A\leq \psi-1=P(P-1)$.
Since the maximum degree of $\mathbf{A}$ is $P(P-1)$, which has an order $\mathcal{O}(P^2)$, we prefer to use $P$ as a reference to determine some other parameters, including $\maxGMRES$, $\zeta$, and $\epsilon_D$.

According to \cite{saad1986gmres}, full GMRES has a complexity order $\mathcal{O}\left(( J+|\mathcal{D}| )JMN\right)$, where $|\mathcal{D}|$ is the number of nonzero elements in each column of $\mathbf{A}$ and $J$ is the number of iterations to reach the tolerance $\epsilon_g$.
Typically, $J$ is less than $|\mathcal{D}|$, and thus the complexity order is roughly $\mathcal{O}\left(|\mathcal{D}|JMN\right)$.
For restarted GMRES, the complexity order reduces to $\mathcal{O}\left( n_{cyc} (\maxGMRES+|\mathcal{D}|)\maxGMRES MN \right)$ in total, where $n_{cyc}$ denotes the required number of restart cycles.
Since generally $\maxGMRES\ll |\mathcal{D}|$, restarted GMRES has a complexity order $\mathcal{O}\left( n_{cyc} |\mathcal{D}|\maxGMRES MN \right)$.

In FSPAI, the computational complexity is dominated by computing the inverse $\mathbf{A}(\widetilde{\setj}_{\fidx}, \widetilde{\setj}_{\fidx})^{-1}$ of (\ref{l1}).
Since FSPAI aims to derive an approximation of the Cholesky factor of $\mathbf{A}^{-1}$, we focus on the sparsity level of $\mathbf{L}$.
According to Algorithm \ref{FSPAIalgo}, if a column of $\mathbf{L}$ has a degree $\varphi_L=D$, the complexity order of this column is $\mathcal{O}(D^4)$.
However, the practical degree $D$ is quite problem-dependent, because the random matrix $\mathbf{A}$ may lead to strikingly different sparsity level of $\mathbf{L}$.
Nevertheless, with the aid of $F_{\varphi_L}(D)$, the average complexity is derived as $\mathcal{O}(\sum_{D=1}^{\zeta}D^4P(\varphi_L=D)MN/n_p)=\mathcal{O}(\mathbb{E}\{\varphi_L^4\}MN/n_p)$, where $n_p$ is the number of processors available for parallel computing and $\zeta$ is the maximum node degrees of $\mathbf{L}$ as given in Algorithm \ref{FSPAIalgo}.
Furthermore, the worst case of FSPAI has a complexity order $\mathcal{O}(\zeta^4MN/n_p)$, where each node has degree $\zeta$.
Thus, even in the worst case, the complexity increases linearly with $MN$.
As will be observed in the simulations, $\zeta=P $ is sufficient for our equalizer, and $\varphi_L$ typically concentrates around zero, especially at high $E_b/N_0$.
This indicates that $\mathbf{L}$ is generally very sparse, and the average complexity is far less than that of the worst case.

\vspace{-0pt}
\section{Simulation Results} \label{section7}
In this section, the BER performance of OTFS with different detectors are simulated and discussed.
Then, we analyze the convergence rate of full and restarted GMRES respectively.
Thereafter, to evaluate the computational complexity of FSPAI, we investigate the sparsity level of $\mathbf{L}$.
Lastly, we apply our proposed equalizer to the fractional Doppler case, and present the corresponding BER performance and complexity.
We set $M=64$, $N=32$, and the length of CP is $L=16$.
An $R_c=0.5$ convolutional code with generator (5,7) in octal is used.
It is assumed that each path gain follows Rayleigh distribution with uniform power delay profile, i.e., $\mathbb{E}\{|h_p|^2\}=1/P$.
The delay and Doppler indices are uniformly chosen such that $0\leq l \leq l_{\mathrm{max}}$ and $-k_{\mathrm{max}}\leq k \leq k_{\mathrm{max}}$, where $l_{\mathrm{max}}=10$ and $k_{\mathrm{max}}=6$.
Compared with the DI-S-MMSE turbo equalizer, the DI-S-MMSE equalizer is used for uncoded OTFS systems, where the MMSE estimator treats its own extrinsic information generated previously as the {\it a priori} information in the next outer iteration.

In \figurename~\ref{BERR} (a), the BER performance of OTFS system with our proposed DI-S-MMSE (turbo) receivers is presented.
For comparison, the OTFS systems with LMMSE, MPA \cite{raviteja2018interference} and (coded) SPA \cite{kschischang2001factor} detectors are also simulated.
Here, the SPA is a near-optimal symbol-wise MAP algorithm \cite{9439819}, and considering that the complexity of SPA detector increases exponentially with $P$, we choose $P=4$.
At the first outer iteration, the DI-S-MMSE equalizer resorts to GMRES and achieves the same BER of the conventional LMMSE equalizer.
After 5 outer iterations, the BER of the DI-S-MMSE equalizer has a roughly $4$ dB gain compared with that of the LMMSE detector at around $\text{BER}=10^{-3}$, and roughly 1.3 dB gain compared with that of the MPA detector.
More importantly, it approaches the performance of the SPA detector.
Furthermore, with the help of the decoder, the DI-S-MMSE turbo equalizer achieves around 2 dB code gain compared with the DI-S-MMSE equalizer at BER $= 10^{-4}$.
Similarly, the BER performance of the DI-S-MMSE turbo equalizer approaches that of coded SPA detector.

\begin{figure*}[t]

  \centering
  \subfigure[]{
  \begin{minipage}[t]{0.45\linewidth}
    \centering
    \includegraphics[width=3in]{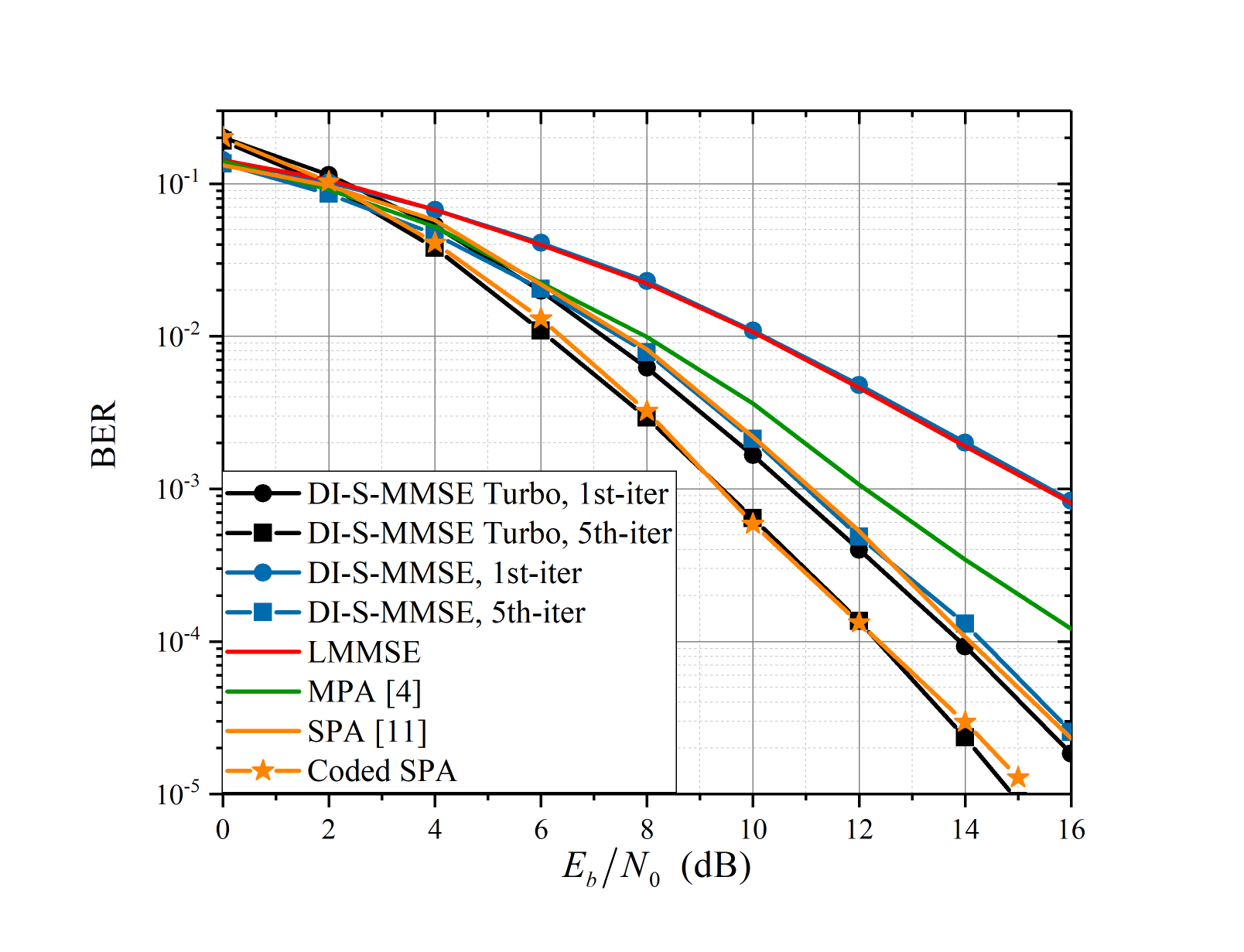}
    %\caption{fig1}
  \end{minipage}
  \label{BER1}
  }%
  \subfigure[]{
    \begin{minipage}[t]{0.5\linewidth}
      \centering
      \includegraphics[width=3in]{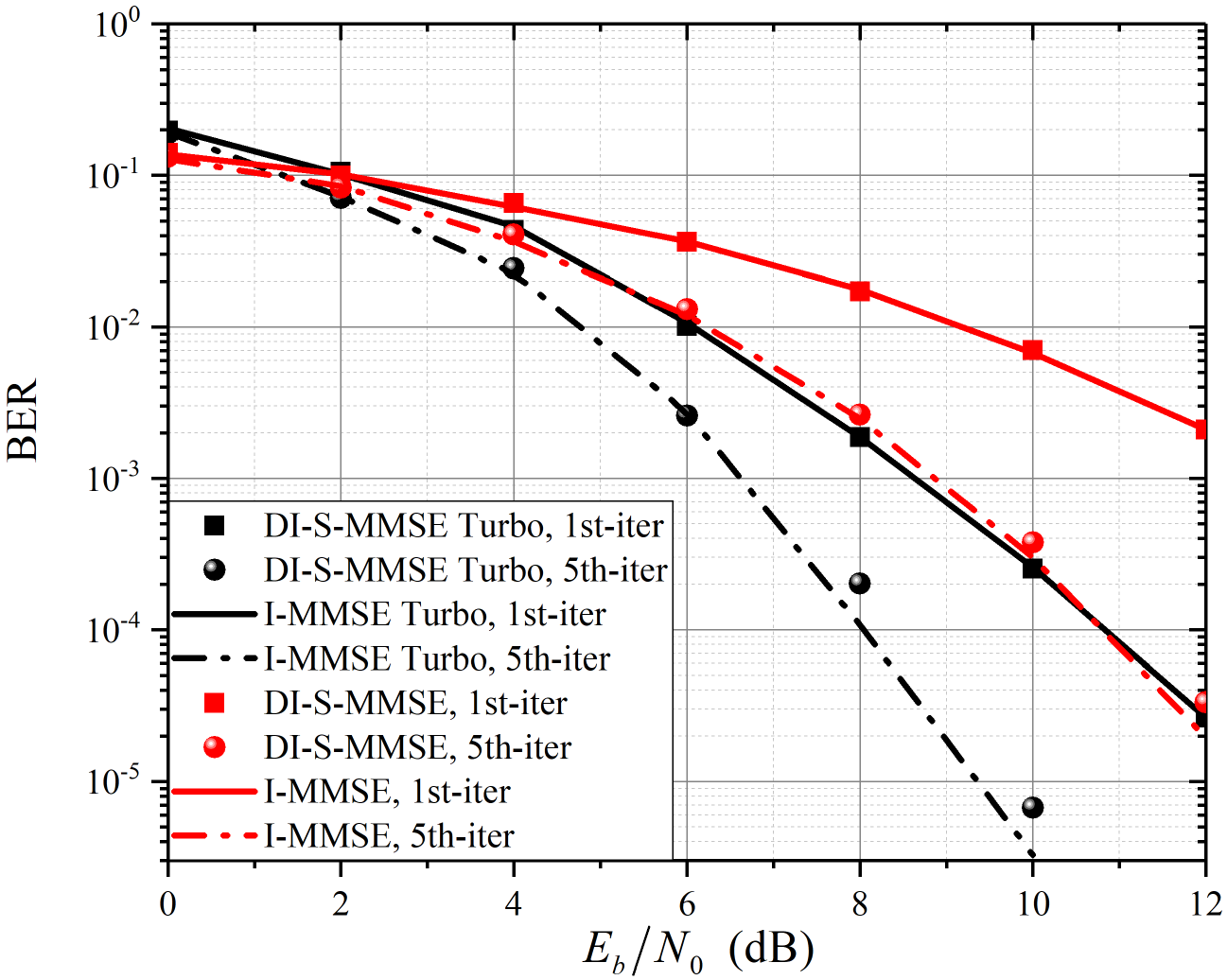}
      %\caption{fig2}
    \end{minipage}
    \label{BER2}
  }%

  \centering
  \caption{The BER performance of our proposed DI-S-MMSE (turbo) equalizers, where $\epsilon_g=\epsilon_f=\epsilon_A=10^{-3}$, $\epsilon_D=P/4$ and $\zeta=P$.
  (a) With $P=4$, we present the BER performance of the LMMSE equalizer, MPA detector and (coded) SPA detector for comparison. 
  (b) With $P=8$, we simulate the BER performance of the I-MMSE (turbo) receivers as a lower bound, where $\mathbf{A}^{-1}$ is directly computed.} 
  \label{BERR}
  \vspace{-5pt}
\end{figure*}

\figurename~\ref{BERR} (b) shows the BER performance of our proposed DI-S-MMSE (turbo) receivers with $P=8$.
To evaluate the performance loss, we present the BER performance of iterative MMSE (I-MMSE) and I-MMSE turbo receivers, where $\mathbf{A}^{-1}$ is directly computed without resorting to GMRES and FSPAI.
Thus, the I-MMSE (turbo) receivers serve as lower bounds of our proposed receivers.
It is found that there is no BER performance loss at the first outer iteration.
On the other hand, at the subsequent outer iterations, the use of FSPAI brings a marginal BER performance degradation.
For example, there is around $0.2$ dB $E_b/N_0$ loss after 5 outer iterations.
Thus, $\zeta=P$ is sufficient for the DI-S-MMSE (turbo) equalizer.

\figurename~\ref{GMRES} shows the convergence performance of full and restarted GMRES.
In \figurename~\ref{GMRES} (a), we investigate the average relative residual norms of full GMRES applied to two sparse linear systems and their upper bounds with the number of inner iterations, where GMRES1 and GMRES2 refer to linear systems $\mathbf{A}\mathbf{f}_1=\mathbf{{y}}$ and $\mathbf{A}\mathbf{f}_2=\mathbf{h}_n$ respectively.
It is observed that full GMRES can converge steadily at any $E_b/N_0$ in both sparse linear systems.
In addition, the average relative residual norms of the upper bound have a gap compared with those of both linear systems, but they have the same slope.
This clearly substantiates Theorem \ref{thm1}, where the upper bound provides a good approximation of the convergence rate.
Furthermore, GMRES converges faster at a smaller $E_b/N_0$.
For example, for GMRES1 and a given tolerance $10^{-3}$, around 15 inner iterations are required at $E_b/N_0=6$ dB, while around 22 inner iterations are needed at $E_b/N_0=10$ dB.
This verifies the observation in Corollary \ref{prop1}.

\begin{figure*}[t]
  \centering
  \subfigure[Full GMRES]{
  \begin{minipage}[t]{0.45\linewidth}
    \centering
    \includegraphics[width=3in]{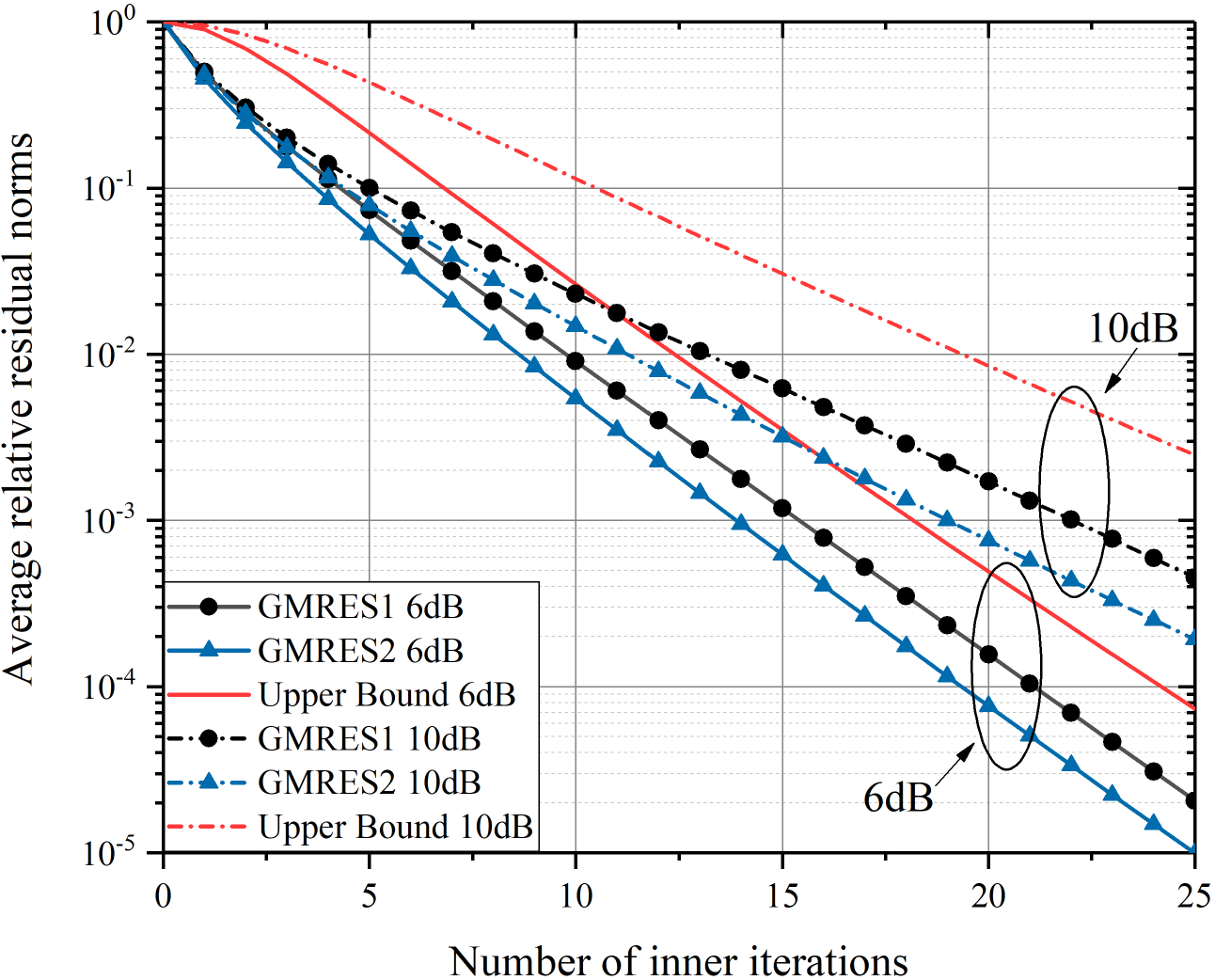}
    %\caption{fig1}
  \end{minipage}
  \label{NoRestart} 
  }%
  \subfigure[Restarted GMRES]{
    \begin{minipage}[t]{0.5\linewidth}
      \centering
      \includegraphics[width=3in]{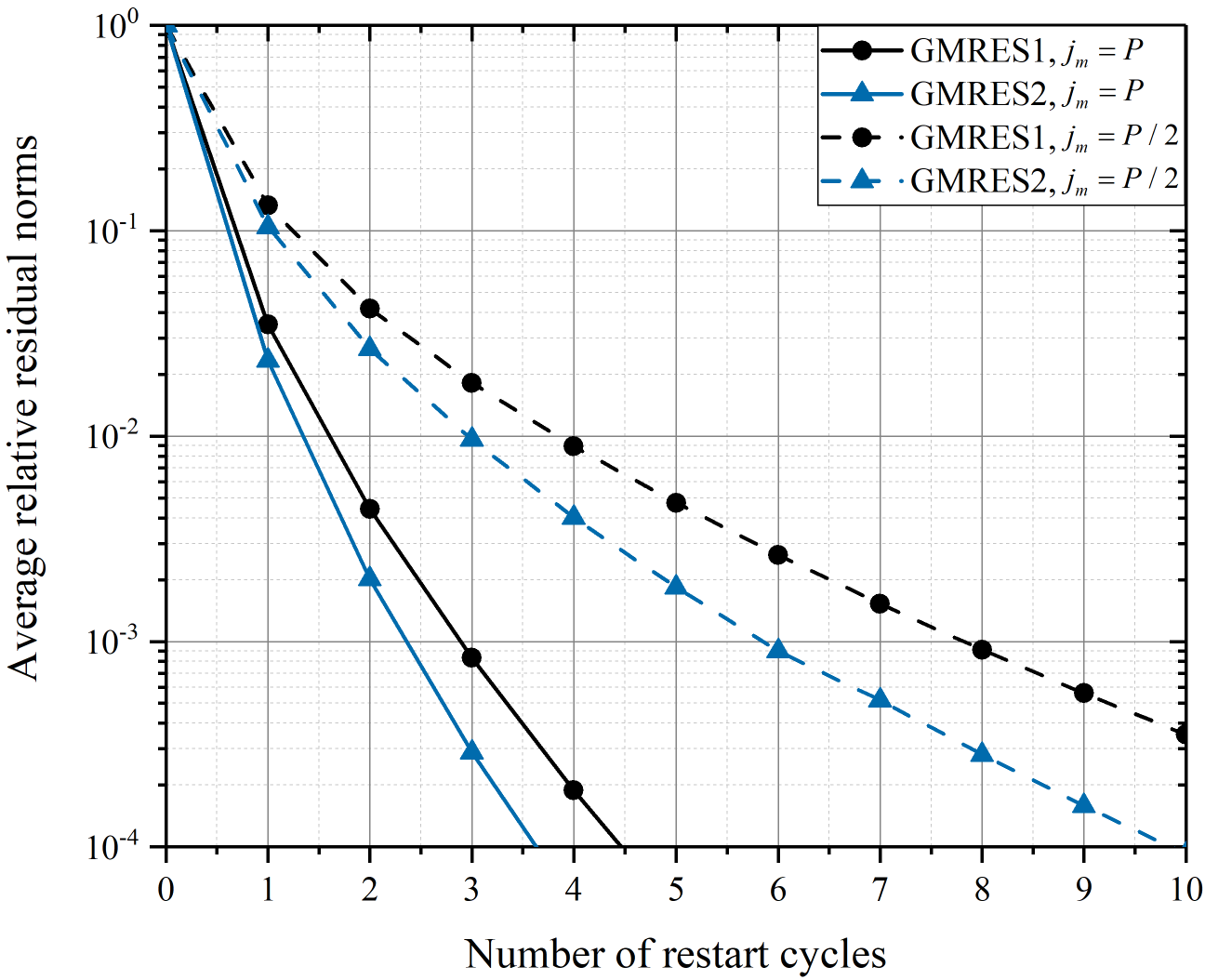}
      %\caption{fig2}
    \end{minipage}
    \label{Restart}
  }%

  \centering
  \caption{Average relative residual norms of full and restarted GMRES applied to two linear systems, where $P=8$.
    (a) Full GMRES cases and the corresponding upper bounds with the number of inner iterations at $E_b/N_0=6$ dB and $10$ dB. (b) Restarted GMRES cases with the number of restart cycles, where $E_b/N_0=10$ dB.} 
  \label{GMRES}
  \vspace{-5pt}
\end{figure*}

\figurename~\ref{GMRES} (b) shows the average relative residual norms of restarted GMRES applied to two sparse linear systems with the number of restart cycles at $E_b/N_0=10$ dB, where the number of inner iterations of a restart cycle is set as $j_m=P=8$ and $j_m=P/2=4$ respectively.
We observe that the convergence of the current restart cycle tends to be slower than that of the previous one.
For example, for GMRES1 with $\maxGMRES=P/2$, the average relative residual norms are reduced by around 90\% (form 1 to around 0.1) at the first restart cycle, but around 70\% (from 0.1 to 0.03) at the second restart cycle.
Hence, the reduction of the residual norms of the current restart cycle is worse than that of the previous restart cycle, which is consistent with the convergence analysis.
Furthermore, a smaller $j_m$ results in a slower convergence, and a larger $n_{cyc}$ is required to reach the drop tolerance $\epsilon_g$.
For both linear systems with $\maxGMRES=P$, $n_{cyc}=3$ is sufficient to reduce the residual norms to $10^{-3}$, while $n_{cyc}=8$ is sufficient for $\maxGMRES=P/2$.
Nevertheless, the latter greatly reduces the complexity of each restart cycle compared with the former, since the maximum dimension of Krylov subspace is restricted.
More importantly, regardless of the choice of $\maxGMRES$, restarted GMRES can always converge, which demonstrates Theorem \ref{thm2}.

\begin{figure*}[t]
  \centering
  \subfigure[]{
  \begin{minipage}[t]{0.45\linewidth}
    \centering
    \includegraphics[width=3in]{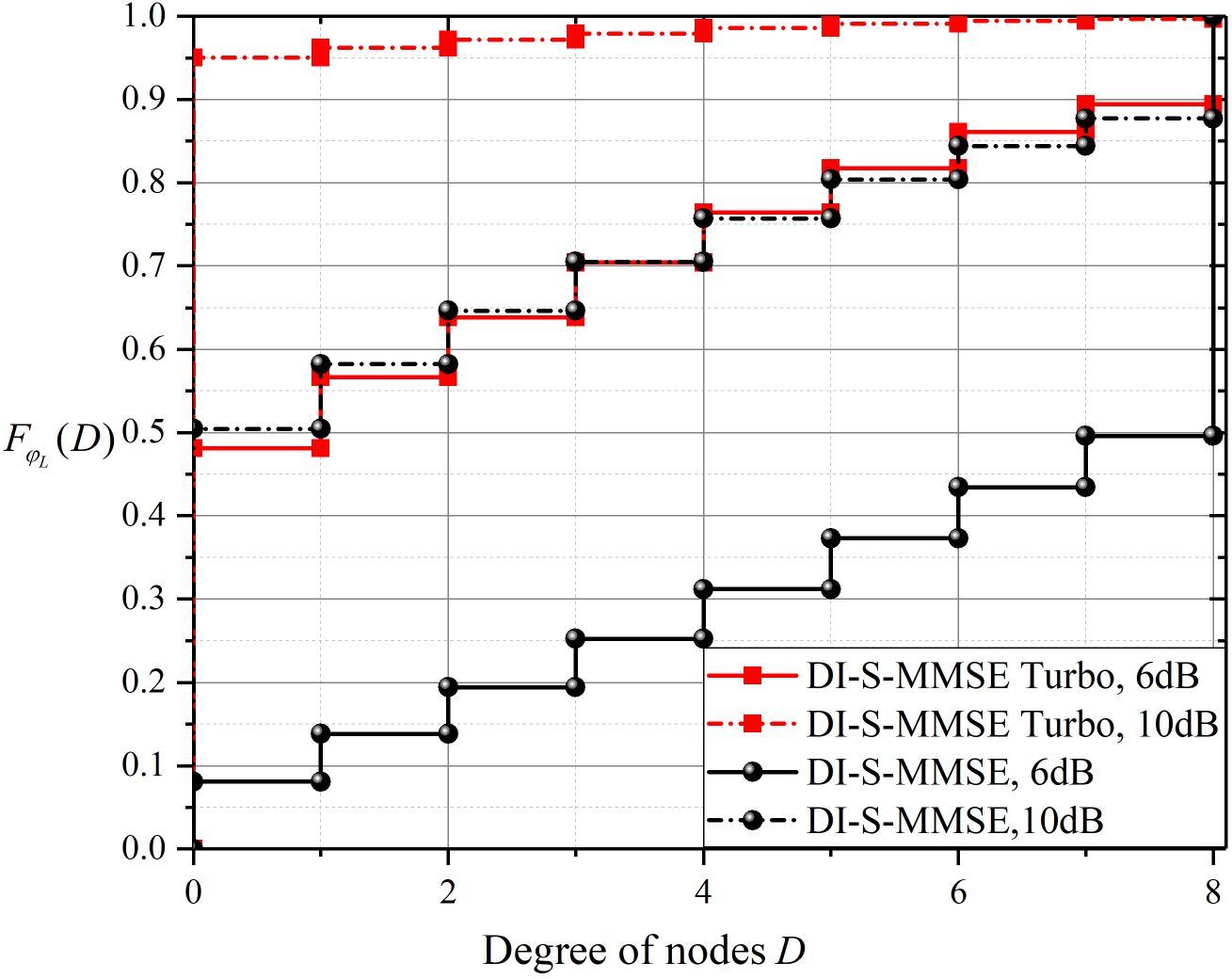}
    %\caption{fig1}
  \end{minipage}
  \label{FSPAI1}
  }%
  \subfigure[]{
    \begin{minipage}[t]{0.5\linewidth}
      \centering
      \includegraphics[width=3in]{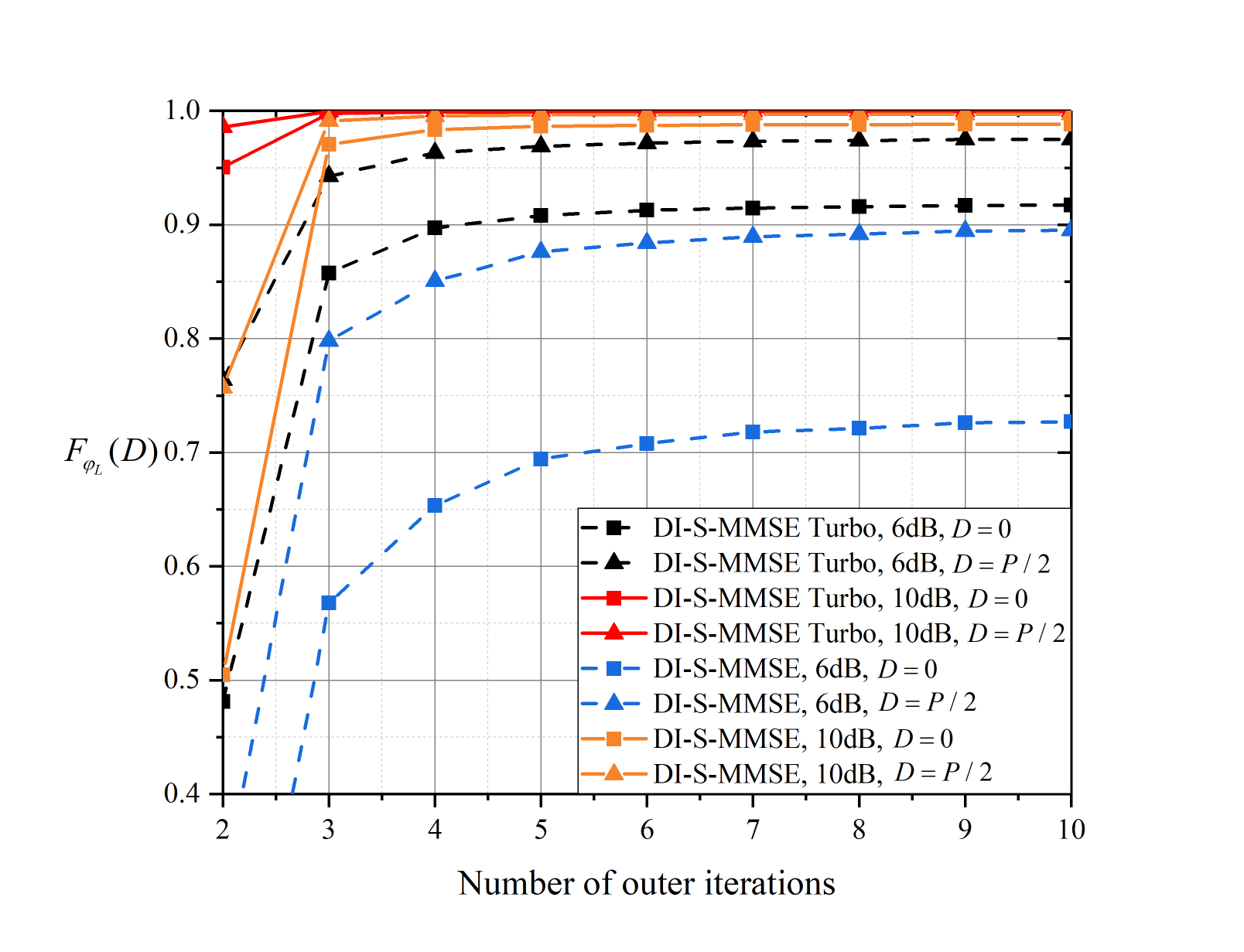}
      %\caption{fig2}
    \end{minipage}
    \label{FSPAI2}
  }%

  \centering
  \caption{The sparsity level of $\mathbf{L}$, measured by CDF $F_{\varphi_L}(D)$, where $E_b/N_0 = 6$ dB and $10$ dB with $P = 8$.
  (a) The sparsity level of $\mathbf{L}$ at the second outer iteration. 
  (b) The sparsity level of $\mathbf{L}$ versus the number of outer iterations, and $D$ is set to be $0$ and $P/2$ respectively.} 
  \label{FSPAI}
  \vspace{-5pt}
\end{figure*}

Then, \figurename~\ref{FSPAI} shows the sparsity level of $\mathbf{L}$, which largely determines the complexity of FSPAI as analyzed in Section \ref{compl}.
From $F_{\varphi_L}(D)$, we can readily observe the distribution of $\varphi_L$, and derive the average complexity of FSPAI.
\figurename~\ref{FSPAI} (a) illustrates the CDF $F_{\varphi_L}(D)$ of the second outer iteration at different $E_b/N_0$ with $P=8$.
After considerable simulations, we have optimized the parameters to achieve a low complexity order and maintain a good BER performance. 
The drop tolerances of GMRES and FSPAI are set to be $\epsilon_g=\epsilon_f=10^{-3}$.
The parameters of sparsification guidelines are chosen to be $\epsilon_A=10^{-3}$ and $\epsilon_D=P/4$ respectively.
The maximum node degree of FSPAI $\zeta=P=8$ is sufficient to avoid noticeable performance loss.
For the DI-S-MMSE turbo receiver at $E_b/N_0=6$ dB, $F_{\varphi_L}(0)\approx 0.48$ shows that around 48$\%$ nodes have degree equal to $0$;
at $E_b/N_0=10$ dB, $F_{\varphi_L}(0)\approx 1$ indicates that most $\varphi_L$ is zero and $\mathbf{L}$ is nearly a diagonal matrix.
This is because with the increase of $E_b/N_0$, the extrinsic information generated by the decoder is more reliable.
Thus, the new variances are closer to zero, and together with the sparsification guidelines applied to $\mathbf{A}$, a sparser approximation $\mathbf{L}$ is made possible.
Furthermore, at the same $E_b/N_0$, $\mathbf{L}$ of the DI-S-MMSE turbo receiver tends to be much sparser than that of the DI-S-MMSE receiver.
This is demonstrated in the figure that $F_{\varphi_L}(D)$ of the DI-S-MMSE turbo receiver is strictly larger than that of the DI-S-MMSE receiver, so that $\mathbf{L}$ is sparser and requires less computation.
Hence, the simulation results verify that FSPAI can indeed reduce the complexity substantially, and the reduction is more manifest with the help of the decoder.

\figurename~\ref{FSPAI} (b) shows the sparsity level of $\mathbf{L}$ with the number of outer iterations at $E_b/N_0=6$ dB and $10$ dB, where $P=8$, and $D$ is chosen to be $0$ and $P/2$ respectively.
Similarly, $F_{\varphi_L}(D)$ of the DI-S-MMSE turbo receiver is strictly larger than that of the DI-S-MMSE receiver at the same $E_b/N_0$, which indicates that our turbo receiver provides a smaller $\varphi_L$ and thus a lower complexity.
Besides, $F_{\varphi_L}(D)$ increases with $E_b/N_0$.
More importantly, $F_{\varphi_L}(D)$ increases with the number of outer iterations, and the overall complexity of FSPAI is mainly dominated by the second outer iteration.
This verifies that the {\it a priori} information is improved with the outer iteration, and the resultant new variances $v_n$ are closer to zero, leading to a sparser $\mathbf{L}$.
For example, at the third outer iteration and $E_b/N_0=10$ dB, both receivers have $F_{\varphi_L}(0)> 0.96$ and $F_{\varphi_L}(P/2)\approx 1$.
Hence, $\mathbf{L}$ is nearly diagonal, which has 96$\%$ nodes with degree zero and almost all the nodes with degree less than $P/2$.
From \figurename~\ref{BERR} and \figurename~\ref{FSPAI}, it is verified that the turbo structure not only improves the BER performance, but also greatly reduces the complexity.
Therefore, our DI-S-MMSE turbo equalizer can indeed achieve a great performance with a preferable complexity.

\begin{figure*}[t]

  \centering
  \subfigure[]{
  \begin{minipage}[t]{0.45\linewidth}
    \centering
    \includegraphics[width=3in]{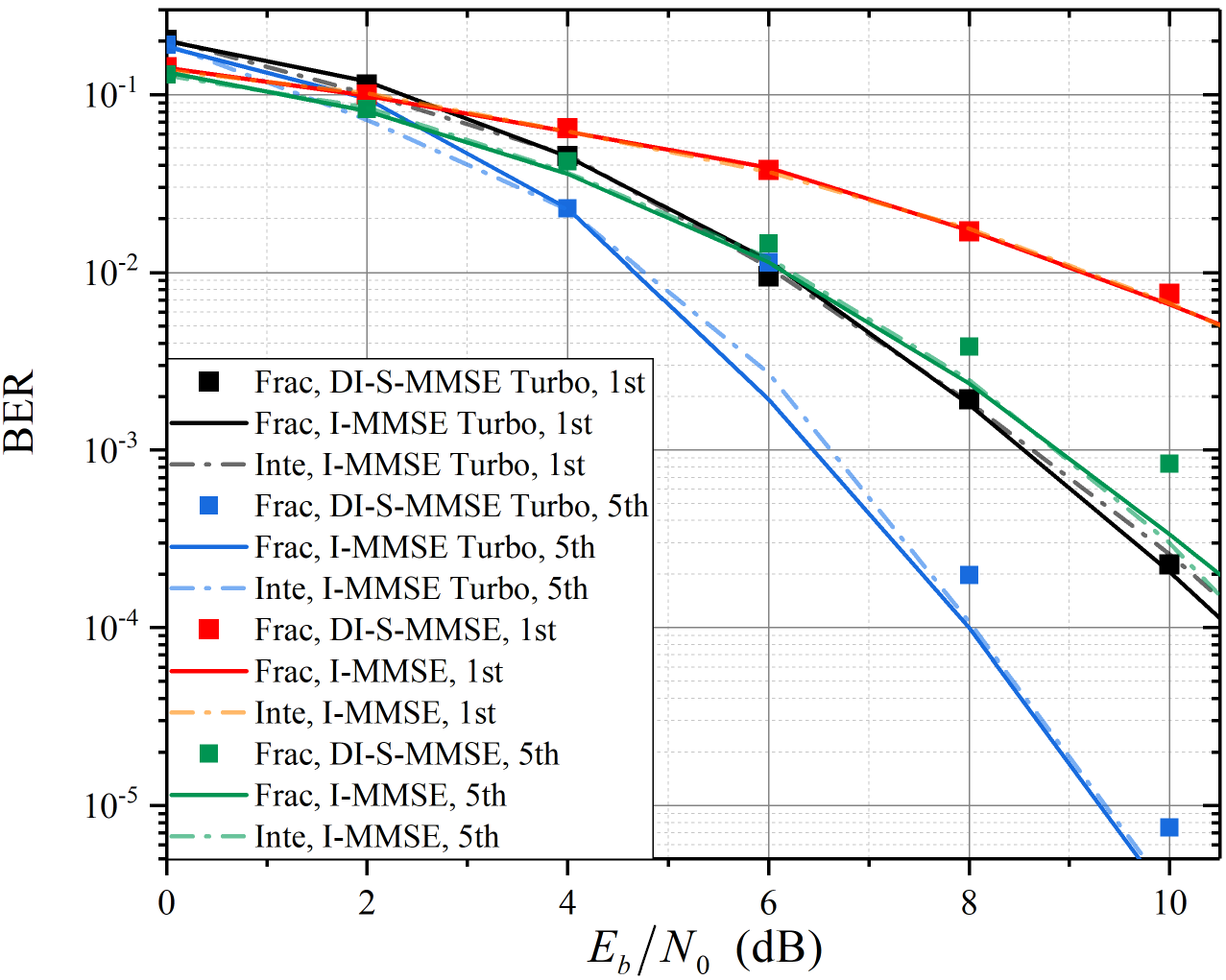}
    %\caption{fig1}
  \end{minipage}
  \label{BER_frac}
  }%
  \subfigure[]{
    \begin{minipage}[t]{0.5\linewidth}
      \centering
      \includegraphics[width=3in]{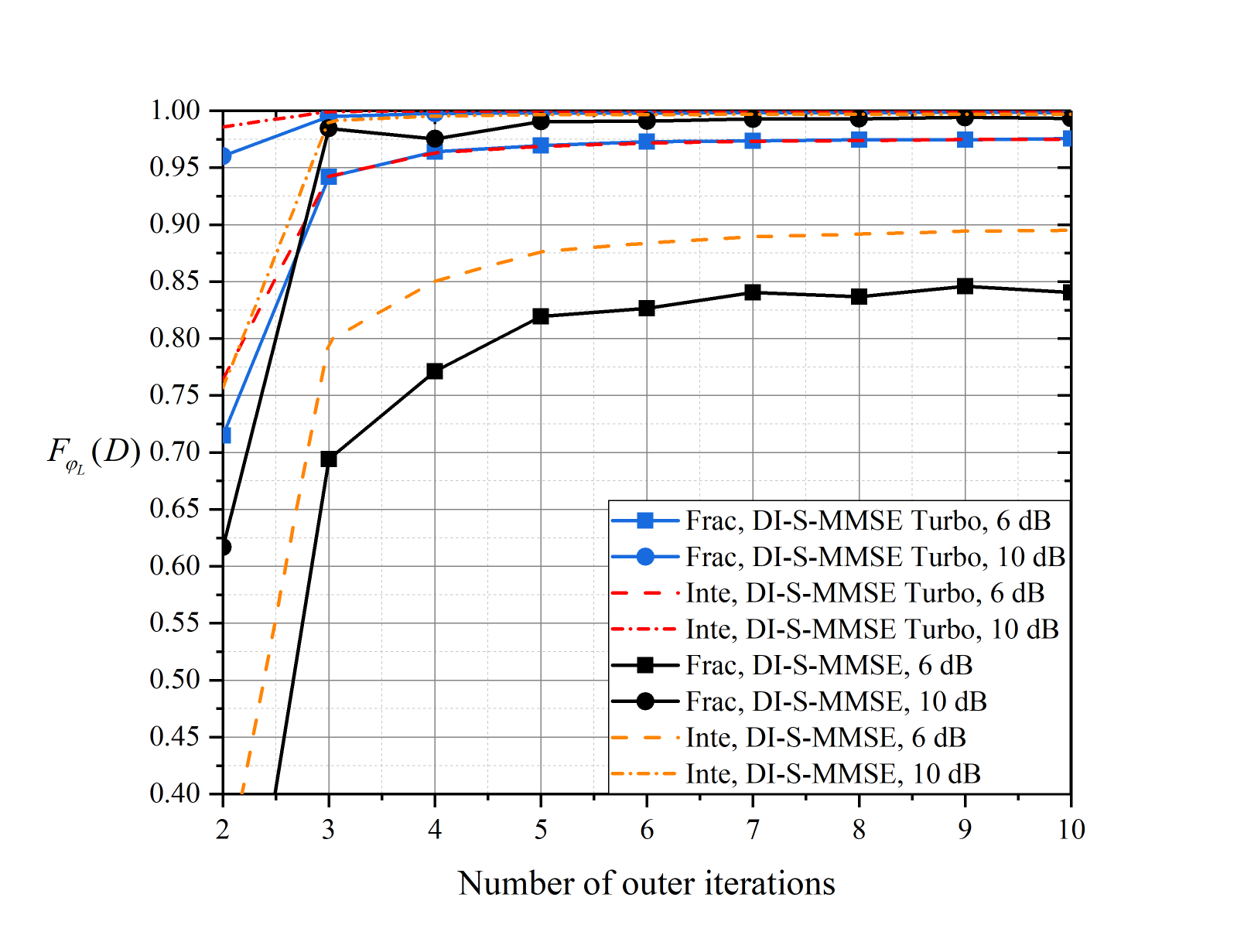}
      %\caption{fig2}
    \end{minipage}
    \label{FSPAI_frac}
  }%

  \centering
  \caption{
  The BER performance and the sparsity level of $\mathbf{L}$ with the fractional Doppler (denoted by "Frac"), where $P=8$, $\epsilon_g=\epsilon_f=\epsilon_A=10^{-3}$, $\epsilon_D=P/4$. 
  For comparison, the BER and sparsity level with integer Doppler shifts (denoted by "Inte") are also presented.
  (a) The BER performance of our proposed DI-S-MMSE (turbo) equalizers. 
  (b) The sparsity level of $\mathbf{L}$ versus the number of outer iterations at $E_b/N_0 = 6$ dB and $10$ dB respectively, and $D$ is set to be $P/2$.}

  \vspace{-5pt}
\end{figure*}

In the following, we apply our equalizer to the fractional Doppler case, and analyze the corresponding BER and complexity respectively.
Here, we use Jakes' formula to generate Doppler taps as $k_p+\kappa_p = k_{\mathrm{max}}\mathrm{cos}(\theta_p)$ for $p=1,\dots,P$, where $k_{\mathrm{max}}=6$ and $\theta_p$ is uniformly distributed over $[-\pi,\pi]$. 
With the help of the decoder in the DI-S-MMSE turbo equalizer, $\zeta=P$ is sufficient to guarantee a good BER performance, while $\zeta=3P$ is required for our uncoded equalizer.
\figurename~\ref{BER_frac} shows the BER with $P=8$, where the solid lines and square points denote the BER of the I-MMSE (turbo) and DI-S-MMSE (turbo) equalizers with the fractional Doppler respectively.
Here, the former serves as the lower bound of the latter.
We observe that there is no BER performance loss at the first outer iteration, which verifies the effectiveness of GMRES.
However, even with $\zeta=3P$ for our uncoded equalizer, around 0.8 dB $E_b/N_0$ loss is incurred at BER = $10^{-3}$ due to the absence of the decoder.
In contrast, only $\zeta=P$ is sufficient for the DI-S-MMSE turbo equalizer and only causes around $0.3$ dB performance loss at BER = $10^{-4}$. 
This is because the decoder can output highly reliable extrinsic information, and thus the updated variances are close to zero.
As a result, a highly sparse $\mathbf{A}$ can be derived without too much performance loss. 
More importantly, the dash lines represent the BER of the I-MMSE (turbo) equalizers without the fractional Doppler from \figurename~\ref{BER2}.
From the solid and dash lines, we observe that the I-MMSE (turbo) can achieve the same BER regardless of the fractional Doppler.
This indicates that our equalizer can handle the fractional Doppler well.

In \figurename~\ref{FSPAI_frac}, we present the sparsity level $F_{\varphi_L}(D)$ of $\mathbf{L}$ with the fractional Doppler at $E_b/N_0=6$ dB and $10$ dB.
For comparison, we also present $F_{\varphi_L}(D)$ from \figurename~\ref{FSPAI2} where the fractional Doppler does not exist.
As shown in \figurename~\ref{FSPAI_frac}, $\mathbf{L}$ is much sparser in the DI-S-MMSE turbo equalizer (blue lines) than that in our uncoded equalizer (black lines).  
For example, at the second outer iteration and $E_b/N_0 = 10$ dB, the former has around 96\% nodes with degree not larger than $P/2$, while the latter has around 62\% nodes with degree not larger than $P/2$.
More importantly, compared with the integer Doppler case (red lines), the sparsity level of $\mathbf{L}$ in the DI-S-MMSE turbo equalizer is hardly influenced.
This is because the updated variances derived from the decoder is so close to zero that the first term of $\mathbf{A}=\mathbf{H}_{\mathrm{DD}}\mathbf{V}\mathbf{H}_{\mathrm{DD}}^{\mathrm{H}}+N_0\mathbf{I}$ is almost negligible.
In this case, the diagonal elements of $\mathbf{A}$ have dominant magnitude, and thus $\mathbf{L}$ can still be highly sparse.
This observation validates that the decoder can not only improve the BER performance, but also greatly reduce the complexity, even with the fractional Doppler.

The computational complexity of the aforementioned detectors is summarized in Table \ref{comp}, where $n_{\mathrm{iter}}$ is the number of iterations of the MPA and SPA detector, and also the number of outer iterations of our equalizer.
In our equalizer, we use restarted GMRES and neglect the decoder for a fair comparison.
It is observed that our equalizer has a complexity order linearly increasing with $MN$.
Besides, since typically $\zeta=P$ is sufficient, even in the worst case, the overall complexity order of our equalizer increases with $P^4$ as discussed in Section \ref{compl}.
However, as observed in \figurename~\ref{FSPAI}, the average complexity is far less than the worst case.
In comparison with the LMMSE equalizer in \cite{tiwari2019low} which has a log-linear order of complexity, our equalizer not only achieves a better BER performance, but also has an even lower complexity order.
The MPA detector has a lower complexity than our equalizer but sacrifices the BER performance as shown in \figurename~\ref{BER1}.
For the SPA detector, the complexity order of the SPA detector is exponential to $P$.
Hence, in richly scattered environment with large $P$, our equalizer has much lower complexity.

\vspace{0pt}
\begin{table*}[t] 
  \small \label{comp}
  \begin{spacing}{1.4}
  \centering
  \caption{Computational complexity of different detectors}
  \begin{tabular} {|c|c|}
    \hline
    Low complexity LMMSE equalizer \cite{tiwari2019low} & $\mathcal{O}(\frac{MN}{2}\log _2M+MN(2l_{max}^2+2P^2k_{max}))$                                              \\
    \hline
    MPA detector \cite{raviteja2018interference} & $\mathcal{O}(n_{\mathrm{iter}}MN P |S| ))$                                              \\
    \hline
    SPA detector\cite{9439819}                          & $\mathcal{O}(n_{\mathrm{iter}}MNP|\mathcal{S}|^{P})$                                                        \\
    \hline
    Our proposed equalizer                              & $\mathcal{O}( MN(n_{cyc} |\mathcal{D}|\maxGMRES+\frac{n_{\mathrm{iter}-1}} {n_p} \mathbb{E}\{\varphi_L^4\}))$ \\
    \hline 
  \end{tabular}
  \vspace{-0pt}
\end{spacing}
\end{table*}

\vspace{0pt}
\section{Conclusion} \label{section8}
This paper developed a DI-S-MMSE turbo equalizer.
To exploit the inherent sparsity of the channel matrix, we resorted to graph theory to analyze the sparse matrices and proposed two sparsification guidelines.
At the initial outer iteration, we applied GMRES to avoid directly computing the inverse of the covariance matrix with little performance loss.
In addition, GMRES can be readily extended to other LMMSE systems if the covariance matrix is sparse.
Then, we modified FSPAI to derive an approximate inverse of the covariance matrix.
Next, we proved the global convergence of full and restarted GMRES respectively, and demonstrated that the overall complexity order of our proposed equalizer increases linearly with the frame size.
Simulation results demonstrated that our proposed equalizer can indeed deliver a great BER performance with a tremendously reduced complexity and regardless of the fractional Doppler.

\vspace{-0pt}
\appendix \label{appen}
\section*{Proof of Theorem \ref{thm2}}
After restarted GMRES performs $\maxGMRES$ iterations, i.e., a restart cycle, the residual is given by 
\vspace{-0pt}
\begin{equation}   \label{LS}
  \left\| \mathbf{r}_{\maxGMRES} \right\|=\left\| \mathbf{r}_0-\mathbf{A}\Delta{ \mathbf{f  }}_{\maxGMRES}' \right\|=
  \min \limits_{\mathbf{\Delta \mathbf{f  }}_{\maxGMRES}\in \mathcal{K}_{\maxGMRES}} \left\| \mathbf{r}_0-\mathbf{A}\Delta \mathbf{f  }_{\maxGMRES} \right\|,
  \vspace{-0pt}
\end{equation}
where $\Delta{ \mathbf{f  }}_{\maxGMRES}'$ is the solution of the LS problem.
Considering $\Delta{ \mathbf{f  }}_{\maxGMRES}'\in\mathcal{K}_{\maxGMRES}$, $\mathbf{A}\Delta{ \mathbf{f  }}_{\maxGMRES}'$ belongs to a subspace spanned by $\{ \mathbf{Ar}_0,\dots,\mathbf{A}^{\maxGMRES}\mathbf{r}_0 \}$.
Then, from (\ref{LS}), the minimum residual is obtained when $ \mathbf{r}_{\maxGMRES} $ is perpendicular to this subspace and $\mathbf{A}\Delta{ \mathbf{f  }}_{\maxGMRES}'$ is the projection of $\mathbf{r}_0$ onto this subspace.
Hence, $\|\mathbf{r}_{\maxGMRES}\|^2 = \|\mathbf{r}_0\|^2 - \|\mathbf{A}\Delta{ \mathbf{f  }}_{\maxGMRES}'\|^2$, and $\mathbf{A}\Delta{ \mathbf{f  }}_{\maxGMRES}'=\mathbf{K}(\mathbf{K}^{\mathrm{H}}\mathbf{K})^{-1}\mathbf{K}^{\mathrm{H}}\mathbf{r}_0$ where $\mathbf{K}=[\mathbf{Ar}_0,\dots,\mathbf{A}^{\maxGMRES}\mathbf{r}_0]$ \cite{zitko2000generalization}. 
Now, we have the relative residual
\vspace{-0pt}
\begin{equation}
  \begin{aligned}  \label{rat}
  &\|\mathbf{r}_{\maxGMRES}\|^2 / \|\mathbf{r}_{0}\|^2
  = 1 - \|\overline{\mathbf{K}}(\overline{\mathbf{K}}^{\mathrm{H}}\overline{\mathbf{K}})^{-1}\overline{\mathbf{K}}^{\mathrm{H}}\overline{\mathbf{r}}_0\|^2\\
  =&1-\overline{\mathbf{r}}_0^{\mathrm{H}} \overline{\mathbf{K}}(\overline{\mathbf{K}}^{\mathrm{H}}\overline{\mathbf{K}})^{-1}\overline{\mathbf{K}}^{\mathrm{H}} \overline{\mathbf{r}}_0,
  \vspace{-0pt}
  \end{aligned}
\end{equation}
where $\overline{\mathbf{r}}_0=\mathbf{r}_0/\|\mathbf{r}_0\|$ and $\overline{\mathbf{K}}=[\mathbf{A\overline{r}}_0,\dots,\mathbf{A}^{\maxGMRES}\mathbf{\overline{r}}_0]$.
Then, the upper bound of (\ref{rat}) is given by 
  \vspace{-0pt}
  \begin{equation}
   \label{restartG}
   \begin{aligned}
    \frac{\|\mathbf{r}_{\maxGMRES}\|^2} {\|\mathbf{r}_{0}\|^2}
    &\leq 1 - \overline{\mathbf{r}}_0^{\mathrm{H}} \overline{\mathbf{K}}\overline{\mathbf{K}}^{\mathrm{H}} \overline{\mathbf{r}}_0/ \lambda_{\mathrm{max}}(\overline{\mathbf{K}}^{\mathrm{H}}\overline{\mathbf{K}}) \\
    &\overset{(a)}{\leq} 1 - \sum_{j=1}^{\maxGMRES}\left|\overline{\mathbf{r}}_0^{\mathrm{H}}\mathbf{A}^j \overline{\mathbf{r}}_0\right|^2 / \mathrm{tr}(\overline{\mathbf{K}}^{\mathrm{H}}\overline{\mathbf{K}})  
    = 1-\varepsilon_{\maxGMRES},
   \end{aligned}
   \vspace{-0pt}
  \end{equation}
  where (a) follows from the fact that $\overline{\mathbf{K}}^{\mathrm{H}}\overline{\mathbf{K}}$ is a Gram matrix and thus positive semidefinite.
  It is easy to know that $\mathrm{tr}(\overline{\mathbf{K}}^{\mathrm{H}}\overline{\mathbf{K}})>0$.
  On the other hand, since $\mathbf{A}^j$ is HPD, $\overline{\mathbf{r}}_0^{\mathrm{H}}\mathbf{A}^j \overline{\mathbf{r}}_0>0$ for $j\geq 1$ and any $\overline{\mathbf{r}}_0$.
  Hence, for any $\maxGMRES \geq 1$, $\varepsilon_{\maxGMRES}$ is always positive, and thus the relative residual is less than one.
  This proves the global convergence of restarted GMRES.

\vspace{10pt}
\bibliographystyle{IEEEtran}
\bibliography{IEEEabrv,reference}

\begin{IEEEbiography}[{\includegraphics[width=1in,height=1.25in,clip,keepaspectratio]{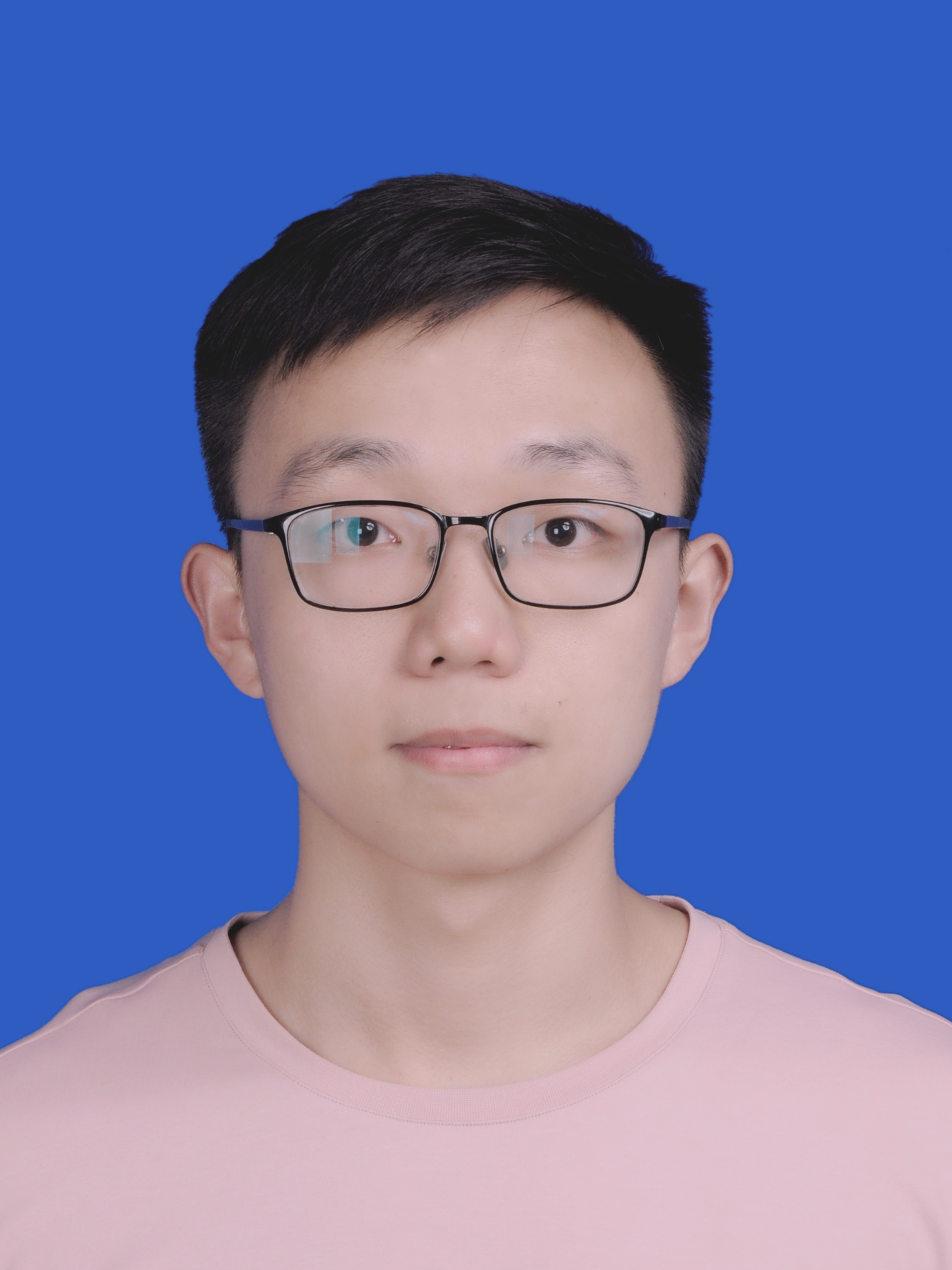}}]{Haotian Yu}
(Student Member, IEEE) received the B.E. degree in electronic and information engineering from Harbin Institute of Technology, Harbin, China, in 2020, where he is currently pursuing the M.S. degree in electronic and information engineering. His research interests include signal processing and channel coding.
\end{IEEEbiography}

\begin{IEEEbiography}[{\includegraphics[width=1in,height=1.25in,clip,keepaspectratio]{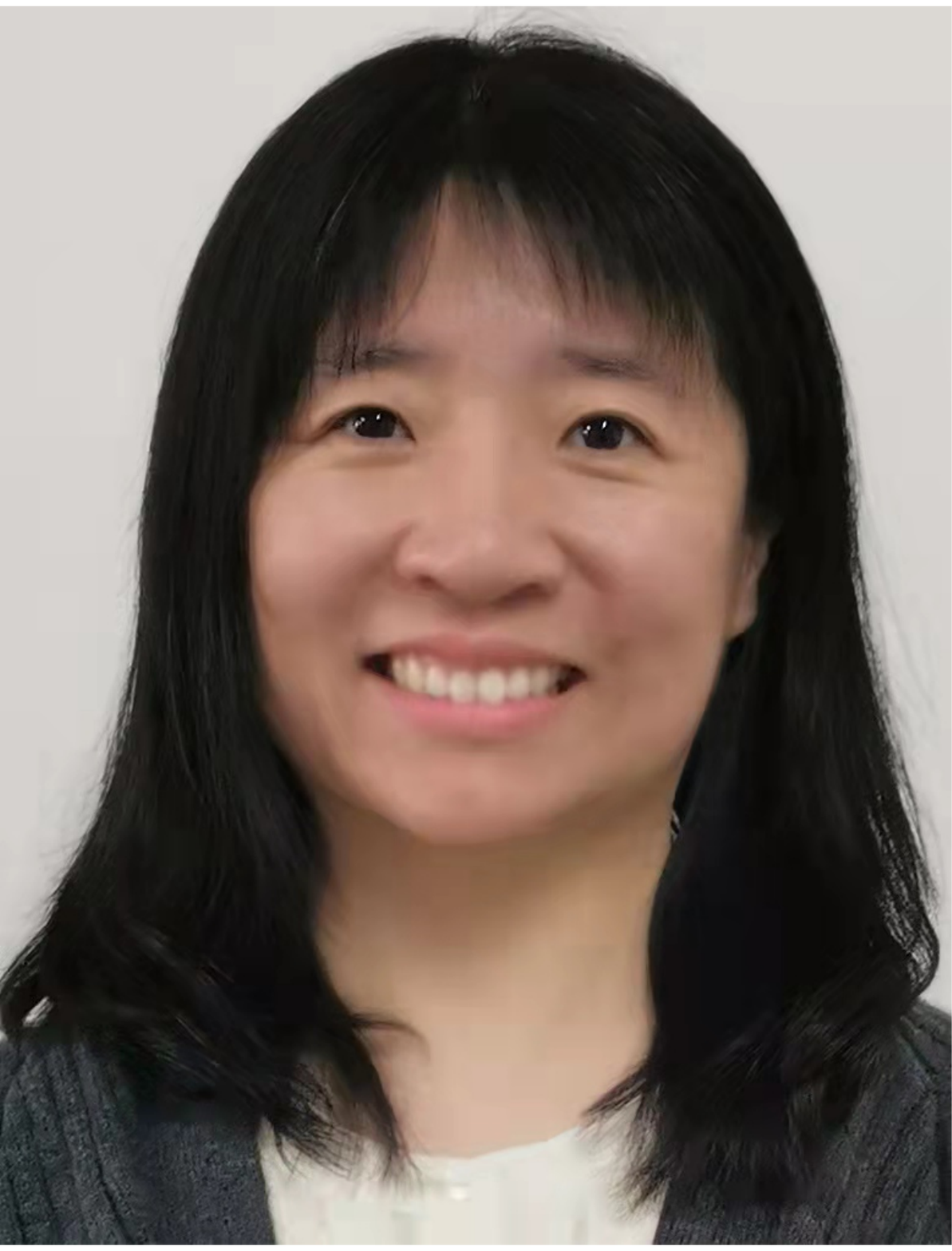}}]{Qiyue Yu}
  (M'08-SM'17) received her B. Eng., M. Eng., and Ph.D. degrees from Harbin Institute of Technology (HIT), China, in 2004, 2006, and 2010, respectively.
  She is currently a full professor at the school of Electronics and Information Engineering, HIT. During Apr. 2007-Mar. 2008, she studied in Adachi Lab, Tohoku University, Japan and was a research assistant of Tohoku University Global COE program. In 2010, she was invited to City University of Hong Kong to research on multi-user MIMO technology. She was invited to University of Southern Queensland in Jan.-Mar. 2014. During Sept. 2015-Sept. 2016, she was a visiting scholar in University of California, Davis, and focused on coding theory. Her research interests include channel codes, multi-access techniques, equalization techniques and MIMO for broadband wireless communications.
\end{IEEEbiography}

\end{document}